\documentclass[a4paper, 10pt]{article}
\title{Generalized weights of convolutional codes}
\date{}
\author{Elisa Gorla and Flavio Salizzoni}

\usepackage{geometry}
\usepackage{amssymb,amsmath,amsthm}
\usepackage{graphicx,url,comment}

\theoremstyle{definition}
\newtheorem{theorem}{Theorem}[section]
\newtheorem{proposition}[theorem]{Proposition}
\newtheorem{lemma}[theorem]{Lemma}
\newtheorem{definition}[theorem]{Definition}
\newtheorem{example}[theorem]{Example}
\newtheorem{remark}[theorem]{Remark}
\newtheorem{corollary}[theorem]{Corollary}
\newtheorem{question}[theorem]{Question}

\newcommand{\F}{\mathbb F}
\newcommand{\Cc}{\mathcal C}
\newcommand{\Aa}{\mathcal A}
\newcommand{\rk}{\mathrm{rk}}
\newcommand{\wt}{\mathrm{wt}}
\newcommand{\maxwt}{\mathrm{maxwt}}
\newcommand{\df}{d_{\mathrm{free}}}
\newcommand{\rev}{\mathrm{rev}}
\newcommand{\supp}{\mathrm{supp}}

\begin{document}

\maketitle
	
\begin{abstract}
In 1997 Rosenthal and York defined generalized Hamming weights for convolutional codes, by regarding a convolutional code as an infinite dimensional linear code endowed with the Hamming metric. In this paper, we propose a new definition of generalized weights of convolutional codes, that takes into account the underlying module structure of the code. We derive the basic properties of our generalized weights and discuss the relation with the previous definition. We establish upper bounds on the weight hierarchy of MDS and MDP codes and show that that, depending on the code parameters, some or all of the generalized weights of MDS codes are determined by the length, rank, and internal degree of the code. We also prove an anticode bound for convolutional codes and define optimal anticodes as the codes which meet the anticode bound. Finally, we classify optimal anticodes and compute their weight hierarchy. 
\end{abstract}

\section{Introduction}

Invariants play an important role in coding theory, both from a theoretical and a practical point of view. They capture qualitative  properties of codes and allow us to quantify the performance of a code with respect to error correction. They also play a role in the classification of codes, as they provide us with an effective way to distinguish non-equivalent codes. 

Generalized weights are among the most studied invariants of codes. Helleseth, Kl\o ve, and Mykkeltveit in~\cite{HKM} propose the first definition of generalized Hamming weights of linear block codes. The $r^{th}$-generalized Hamming weight of a linear block code $\Cc$ is the smallest cardinality of a support of an $r$-dimensional subcode of $\Cc$. Generalized weights are further studied in~\cite{Wei} by Wei, who proves that they characterize the code performance in the wire-tap channel of type II. In fact, Wei shows in~\cite{Wei} that they allow us to measure how much information is gained by an adversary with a given number of taps. The interest in generalized Hamming weights also follows from their connection with the complexity of the minimal trellis diagram, as shown by Forney in~\cite{For}. 

Several definitions of generalized weights in the rank metric have also been proposed. In~\cite{OS} Oggier and Sboui define generalized weights for vector rank-metric codes. In~\cite{KMU} Kurihara, Matsumoto, and Uyematsu introduce relative generalized weights in the same context. In~\cite{Rav16}, Ravagnani proposes a definition of generalized weights for a rank-metric code based on optimal anticodes. In~\cite{MPM} Martinez-Pe\~nas and Matsumoto give different, but related, definitions for generalized weights and relative generalized weights of rank-metric codes. In~\cite{CGLLMS}, the authors propose a definition of weights for sum-rank metric codes that simultaneously extends both the Hamming and the rank metric case.

In~\cite{For}, Forney suggests that generalized Hamming weights may be extended to convolutional codes. Motivated by this observation, Rosenthal and York in~\cite{RY} introduce a notion of generalized Hamming weights for convolutional codes. Given a convolutional code $\Cc$, they forget about its module structure and regard it as an infinite dimensional vector space endowed with the Hamming metric. Their definition of generalized Hamming weights is the natural extension of the usual definition to infinite dimensional codes. Along the same lines, Cardell, Firer, and Napp in a series of papers~\cite{CFN17,CFN19,CFN20} introduce and study a new class of generalized weights for convolutional codes, based on the column distance instead of the Hamming distance. 

The aim of this paper is to introduce a new family of generalized weights for convolutional codes, which takes into account their module structure. Since a convolutional code is an $\F_q[x]$-submodule of $\F_q[x]^n$, it is natural to consider the set of its $\F_q[x]$-submodules instead of that of its $\F_q$-linear subspaces. This point of view was adopted in the past in order to define generalized weights for codes over rings. 
In~\cite{HS}, Horimoto and Shiromoto use submodules to define generalized weights of linear codes over finite chain rings. More recently, in~\cite{GR22} Gorla and Ravagnani give a definition of generalized weights for codes over finite principal ideal rings and for a large class of support functions, by looking at supports of submodules.

After defining generalized weights of convolutional codes, we investigate their basic properties and prove a number of results on generalized weights and related concepts. In particular, we define and study optimal convolutional anticodes.

The paper is organized as follows. In Section 2 we fix the notation and recall some algebraic background and useful facts about convolutional codes. In Section 3 we define the weight of a convolutional code and use this concept to define generalized weights for convolutional codes. We establish their basic properties and prove that our generalized weights are a natural extension of the generalized Hamming weights of linear block codes. We also discuss their relation with the generalized Hamming weights of convolutional codes defined in~\cite{RY}. In Section 4 we show that the computation of our generalized weights can be simplified by considering only submodules with certain properties. For instance, we prove that the generalized weights are realized by submodules generated by codewords of minimal support. In Section 5 we establish some upper bounds for the generalized weights of Maximum Distance Separable (MDS) and Maximum Distance Profile (MDP) codes. Our bounds imply that, depending on the code parameters, some or all of the generalized weights of MDS codes are determined by the code length, rank, and internal degree. Finally, in Section 6, we define the maximum weight of a convolutional code and prove an anticode bound. We define optimal convolutional anticodes as the codes which meet the anticode bound, we classify optimal anticodes, and compute their generalized weights.

\section{Preliminaries}

Let $\F_q$ denote the field with $q$ elements and let $\F_q[x]$ be the ring of univariate polynomials with coefficients in $\F_q$. For $\delta \geq 0$, we denote by $\F_q[x]_{\leq\delta}$ the set of polynomials of degree at most $\delta$. For every positive integers $n$ we write $\F_q[x]^n$ for the direct sum of $n$ copies of $\F_q[x]$. Since $\F_q$ is a field, we have that $\F_q[x]$ is a principal ideal domain. Therefore,
every submodule of $\F_q[x]^n$ is a free $\F_q[x]$-module of finite rank $k\leq n$ that admits a finite basis of cardinality $k$. We recall that given a submodule $M$ of $\F_q[x]^n$
a set $B\subseteq M$ is a \textbf{basis} for $M$ if $B$ generates $M$ and $B$ is \textbf{$\F_q[x]$-linearly independent}, that is, for every subset $\{b_{1},b_{2},\ldots ,b_{n}\}$ of $B$, $r_{1}b_{1}+r_{2}b_{2}+\cdots +r_{n}b_{n}=0$ implies that $r_{1}=r_{2}=\cdots =r_{n}=0$. From here on, we will only work with submodules of $\F_q[x]^n$, so it will be always possible to fix a basis. For $U\subseteq\F_q[x]^n$ a subset, we denote by $\langle U\rangle_{\F_q[x]}=\langle u\mid u\in U\rangle_{\F_q[x]}$ the $\F_q[x]$-module generated by the elements of $U$.

An $(n,k)$ \textbf{convolutional code} $\Cc$ is an $\F_q[x]$-submodule of $\F_q[x]^n$ of rank $k$. We always assume that $\Cc\neq 0$. An element $c(x)\in \Cc$ is an n-tuple $(p_1(x),\dots,p_n(x))$, where
$$p_j(x)=a_{j,0}+a_{j,1}x+\dots+a_{j,s_j}x^{s_j},$$ for all $j\in\{1,\dots, n\}$. Equivalently, we can express $c(x)$ with a more compact notation as
$$c(x)=\sum_{t=0}^{\deg(c(x))}c[t]x^t,$$
where $\deg(c(x))=\max_{j}\deg(p_j(x))$ and $c[t]=(a_{1,t},\dots,a_{n,t})\in\F_q^n$ for all $t$. The {\bf $j^{th}$ truncation} of $c(x)$ is
\begin{equation*}
    c_{[0,j]}(x)=\sum_{t=0}^{j}c[t]x^t.
\end{equation*}
One can associate to $\Cc=\langle c_1,\dots,c_k\rangle_{\F_q[x]}$ the linear block code $$\Cc[0]=\langle c_1[0],\dots,c_k[0]\rangle_{\F_q}.$$
Notice that $\Cc[0]$ does not depend on the choice of a system of generators for $\Cc$ and $\dim(\Cc[0])\leq k$. 

The (Hamming) weight $\wt_H(c)$ of $c\in \F_q^n$ is the number of non zero components of $c$. The \textbf{weight} of an element $c(x)\in\Cc$ is given by		
$$\wt(c(x))=\sum_{k=0}^{\deg(c(x))}\wt_H(c[t]).$$
Let $c(x)=(p_1(x),\dots,p_n(x))$. Then 
$$\supp(c(x))=\{(j,k): a_{j,k}\neq 0\}$$
is the \textbf{support} of $c(x)$. We have $$|\supp(c(x))|=\wt(c(x)).$$ Let $U\subseteq\Cc$ be a subset of $\Cc$. The \textbf{support} of $U$ is $$\supp(U)=\bigcup_{c(x)\in U}\supp(c(x)).$$
If $U=\langle c_1(x),\ldots,c_h(x)\rangle_{\F_q}$ is an $\F_q$-linear space, then it is easy to show that $$\supp(U)=\bigcup_{i=1}^h\supp(c_i(x)),$$
see e.g.~\cite[discussion after Definition 2.8]{GR22}.

If $C=(c[0]\dots c[\deg(c)])\in\F_q^{n\times(\deg(c)+1)}$ is the matrix with columns $c[0],\dots,c[\deg(c)]$, then $\supp(c(x))$ simply corresponds to the support of $C$, i.e., the set of positions of the non zero entries of $C$. Notice that, in this notation, the columns of $C$ are indexed starting from $0$ instead of $1$.

Let $\Cc_1,\Cc_2\in\F_q[x]^n$ be convolutional codes. We call $\Cc_1$ and $\Cc_2$ \textbf{isometric}
if there exists a weight-preserving $\F_q[x]$-isomorphism $\phi:\Cc_1\rightarrow\Cc_2$, that is, $\phi$ is a homomorphism of $\F_q[x]$-modules and $\wt(c(x))=\wt(\phi(c(x)))$ for all $c(x)\in\Cc_1$, see~\cite{Glu}.

A \textbf{generator matrix} of $\Cc$ is a matrix $G(x)$ with entries in $\F_q[x]$ whose rows form a basis of $\Cc$. We denote by $\delta$ the {\bf internal degree} of a convolutional code $\Cc$, i.e., the maximum degree of a full size minor of $G(x)$. It can be shown that $\delta$ is independent of the choice of $G(x)$.
A convolutional code $\Cc\subseteq\F_q[x]^n$ of rank $k$ and degree $\delta$ is an $(n, k,\delta)$ convolutional code. Let $G(x)=(p_{i,j}(x))_{i,j}$ be a generator matrix of $\Cc$ and let $\delta_i=\max_{j=1}^n\deg(p_{i,j}(x))$. Up to a row permutation, we may assume that $\delta_1\geq\delta_2\geq\dots\geq\delta_k$. We say that $G(x)$ is \textbf{row-reduced} if $\delta=\sum_{i=1}^k\delta_i$. Every code $\Cc$ has a row-reduced generator matrix, see e.g.~\cite{Kal}. 

A convolutional code $\Cc$ is \textbf{noncatastrophic} if it has a left-prime generator matrix $G(x)$, i.e., a generator matrix $G(x)$ with the property that if $G(x)=H(x)G'(x)$, then $H(x)$ is unimodular. If $\Cc$ has no left-prime generator matrix, we say that $\Cc$ is {\bf catastrophic}. The following theorem provides some useful  characterizations of being left-prime.

\begin{theorem}[{\cite[Theorem 10.2.6]{LPR}}]\label{theorem:leftprime}
Let $k\leq n$ and let $G(x)\in \F_q[x]^{k\times n}$. 
The following are equivalent:
\begin{enumerate}
\item $G(x)$ is left-prime.
\item The Smith form of $G(x)$ is $(\mathrm{Id}_{k\times k} 0_{k\times (n-k)})$.
\item $G(x)$ admits a right $n\times k$ inverse with entries in $\F_q[x]$.
\end{enumerate}
\end{theorem}

The \textbf{free distance} of a convolutional code $\Cc$ is defined as
$$\df(\Cc)=\min\{\wt(c(x))\mid c(x)\in\Cc\setminus\{0\}\}.$$
In~\cite{RS} Smarandache and Rosenthal established an analogue of the Singleton bound for convolutional codes.

\begin{theorem}[Singleton bound]\label{singletonbound}
Let $\Cc$ be an $(n,k,\delta)$ convolutional code. Then
$$\df(\Cc)\leq(n-k)\left(\left\lfloor \frac{\delta}{k}\right\rfloor+1\right)+\delta +1.$$
\end{theorem}
Similarly to the case of block codes endowed with the Hamming metric, a code that meets the Singleton bound is called \textbf{Maximum Distance Separable} (MDS). Notice that linear block codes are exactly the convolutional codes with internal degree $\delta=0$. Coherently, the Singleton bound from Theorem~\ref{singletonbound} coincides with the usual Singleton bound for $\delta=0$. In particular, a linear block code is MDS if an only if it is MDS when regarded as a convolutional code with internal degree $0$.

The $j^{th}$ \textbf{column distance} of a convolutional code $\Cc$ is defined as
\begin{equation*}
d_j^c(\Cc)=\min\{\wt(c_{[0,j]}(x)):c(x)\in\Cc\text{ and }c[0]\neq0\}.
\end{equation*}
In~\cite{LRS}, the authors proved the following bound.

\begin{theorem}
Let $\Cc$ be an $(n,k,\delta)$ convolutional code. Then
$$d_j^c(\Cc)\leq(n-k)(j+1) +1,$$
for all $j\in\mathbb{N}_0$.
\end{theorem}

A code with $k\neq n$ that achieves this bound for $j=0,\dots,L=\left\lfloor \frac{\delta}{k}\right\rfloor+\left\lfloor \frac{\delta}{n-k}\right\rfloor$ is called \textbf{Maximum Distance Profile} (MDP). 
Finally, $\Cc$ is \textbf{strongly Maximum Distance Separable} (sMDS) if $\df(\Cc)=d^c_{M}(\Cc)$, where $M=\left\lfloor \frac{\delta}{k}\right\rfloor+\left\lceil \frac{\delta}{n-k}\right\rceil$. Strongly MDS convolutional codes are a family of MDS convolutional codes and were introduced in \cite{LRS}.

Let $\Cc$ be an $(n,k,\delta)$ convolutional code. The \textbf{dual} code $\Cc^{\perp}$ is defined as
\begin{equation*}
    \Cc^{\perp}=\{d(x)\in\F_q[x]^n\mid d(x)c(x)^{T}=0\text{ for all }c(x)\in\Cc\}.
\end{equation*}
It is well known that, if $\Cc$ is noncatastrophic, then $\Cc^{\perp}$ is an $(n,n-k,\delta)$ convolutional code and  $(\Cc^{\perp})^{\perp}=\Cc$.

\section{Generalized weights}

In this section, we propose a new definition of generalized weights for convolutional codes and we establish some of their basic properties. We discuss their connection with the generalized weights defined in~\cite{RY} and the generalized weights in the Hamming metric. Finally, we prove by means of an example that these generalized weights do not satisfy Wei duality.
	
We start by giving a characterization of noncatastrophicity, which will be useful in the sequel.

\begin{proposition}\label{proposition:noncatastrophic}
Let $\Cc\subseteq\F_q[x]^n$ be an $(n,k,\delta)$ convolutional code. Then, $\Cc$ is noncatastrophic if and only if for every $r(x)\in\F_q[x]\setminus\{0\}$ we have that
\begin{equation*}
r(x)c(x)\in\Cc\text{ implies }c(x)\in\Cc.
\end{equation*}
\end{proposition}

\begin{proof}
Suppose that $\Cc$ is noncatastrophic and let $G(x)$ be a generator matrix for $\Cc$. Then, by Theorem~\ref{theorem:leftprime} we have that there exists $H\in\F_q[x]^{n\times k}$ such that $GH=\mathrm{Id}$. If $r(x)c(x)\in\Cc$, then there exists $u(x)\in\F_q(x)^k$ such that $u(x)G(x)=r(x)c(x)$. Therefore
$$u(x)=u(x)G(x)H(x)=r(x)c(x)H(x).$$
This implies that $r(x)$ divides each entry of $u(x)$, hence $c(x)\in\Cc$.
	    
In order to prove the converse, we suppose that $\Cc$ is catastrophic and we prove that there exists $r(x)\in\F_q[x]\setminus\{0\}$ and $c(x)\in\F_q[x]^n\setminus \Cc$ such that $r(x)c(x)\in\Cc$. Since $\Cc$ is catastrophic, by Theorem \ref{proposition:noncatastrophic} we may assume up to an isometry that the Smith form of $G(x)$ is $D(x)=(\tilde D\, 0_{k\times (n-k)})$, where $\tilde D$ is a diagonal matrix with entries $d_1(x),\dots,d_k(x)\in\F_q[x]$ such that $\deg(d_1(x))\geq 1$. Let $S,T$ be two invertible matrices such that $SDT=G$. We have that
\begin{equation*}
(1,0,\dots,0)S^{-1}G=(1,0,\dots,0)S^{-1}SDT=(d_1(x),0,\dots,0)T\in\Cc.
\end{equation*}
In order to conclude, it suffices to show that $(1,0,\dots,0)T\notin\Cc$. Suppose by contradiction that there exists $u(x)\in\F_q[x]^k$ such that $u(x)G=(1,0,\dots,0)T$. Let $\tilde u(x)=(\tilde u_1(x),\dots,\tilde u_k(x))=u(x)S$. Since $T$ is invertible, then $\tilde u(x)D=(1,0,\dots,0)$, hence $\tilde u_1(x)d_1(x)=1$. This contradicts the assumption that $\deg(d_1(x))\geq1$. Therefore $(d_1(x),0,\dots,0)T\in\Cc$, while $(1,0,\dots,0)T\notin\Cc$.
\end{proof}

In the next proposition, we collect some facts on $(n,k,0)$ convolutional codes, that we will use throughout the paper.

\begin{proposition}\label{prop:C[0]}
Let $\Cc\subseteq\F_q[x]^n$ be an $(n,k,0)$ convolutional code. Then
$\Cc=\langle\Cc[0]\rangle_{\F_q[x]}$ and $\Cc[0]\subseteq\F_q^n$ is a linear block code with $\dim(\Cc[0])=k$ and minimum distance $d_{\min}(\Cc[0])=\df(\Cc)$. Moreover $\Cc^{\perp}=\langle\Cc[0]^{\perp}\rangle_{\F_q[x]}$.
\end{proposition}

\begin{proof}
Since $\Cc$ is an $(n,k,0)$ convolutional code, then $\Cc[0]=\Cc\cap\F_q^n$ and $\Cc=\langle\Cc[0]\rangle_{\F_q[x]}$. Since $\Cc[0]\subseteq\F_q^n$, then $$\dim(\Cc[0])=\rk(\langle\Cc[0]\rangle_{\F_q[x]})=\rk(\Cc).$$
Since $\Cc[0]\hookrightarrow\Cc$ is an $\F_q$-linear isometry, $d_{\min}(\Cc[0])\geq\df(\Cc)$. However, for any $c(x)\in\Cc$, $c[0]\in\Cc[0]$ and $\wt(c[0])\leq\wt(c(x))$, showing that $d_{\min}(\Cc[0])\leq\df(\Cc)$. Finally, the equality $\Cc^{\perp}=\langle\Cc[0]^{\perp}\rangle_{\F_q[x]}$ follows from the definition of dual code.
\end{proof}

In order to define the generalized weights of a convolutional code, we first wish to define a notion of weight of a code. This replaces the usual notion of cardinality of the support of a code.

\begin{definition}\label{defn:subcodewt}
Let $\Cc\subseteq\F_q[x]^n$ be a convolutional code of rank $k\leq n$. The {\bf weight} of $\Cc$ is $$\wt(\Cc)=\min\{|\supp(\langle c_1(x),\ldots,c_k(x)\rangle_{\F_q})| : \Cc=\langle c_1(x),\ldots,c_k(x)\rangle_{\F_q[x]}\}.$$
\end{definition}

Notice that the set $\{\supp(\langle c_1(x),\ldots,c_k(x)\rangle_{\F_q}) : \Cc=\langle c_1(x),\ldots,c_k(x)\rangle_{\F_q[x]}\}$ may not have a minimum with respect to inclusion, as the next example shows.

\begin{example}
Let $\Cc=\langle (1+x^2,0,1),(1,1,0)\rangle\subseteq\F_2[x]^3$. Then $\Cc$ is a noncatastrophic $(3,2,2)$ binary code and $$\begin{pmatrix} 1+x^2 & 0 & 1 \\ 1 & 1 & 0 \end{pmatrix} \mbox{ and } \begin{pmatrix} 1 & x^2 & 1 \\ 1 & 1 & 0 \end{pmatrix}$$ are two row-reduced generator matrices for $\Cc$ whose $\F_q$-rowspaces have incomparable supports. We claim that these supports are minimal in the set $$\{\supp(c_1(x))\,\cup\,\supp(c_2(x)):\Cc=\langle c_1(x),c_2(x)\rangle_{\F_2[x]}\}.$$ In fact, if $c_1(x),c_2(x)\in\F_2[x]$ are generators of $\Cc$, then $$\supp(c_1(x))\,\cup\,\supp(c_2(x))\supseteq\{(1,0),(2,0),(3,0)\}.$$ Moreover, we have that $\supp(c_1(x))\,\cup\,\supp(c_2(x))\neq\{(1,0),(2,0),(3,0)\}$, otherwise $\langle c_1(x),c_2(x)\rangle_{\F_2}$ would be a two-dimensional binary linear block code of length 3 and the largest dimension of a binary linear block code inside $\Cc$ is one. 
\end{example}

We are now ready to define the generalized weights of a convolutional code.

\begin{definition}\label{defgenw}
Let $\Cc$ be an $(n,k,\delta)$ convolutional code. For $1\leq r\leq k$, the r-th \textbf{generalized weight} of $C$ is
\begin{equation*}
d_r(\Cc)=\min\{\wt(\mathcal{D})\mid \mathcal{D}\subseteq\Cc\mbox{ is a subcode of } \rk(\mathcal{D})\geq r\}.
\end{equation*}
We say that $\mathcal{D}\subseteq\Cc$ {\bf realizes} the $r$-th generalized weight of $\Cc$ if $\rk(\mathcal{D})=r$ and $d_r(\Cc)=\wt(\mathcal{D})$.
\end{definition}

In the next lemma we provide several equivalent formulations of Definition~\ref{defgenw}. In particular, we show that $d_r(\Cc)$ is the minimum weight of a subcode of $\Cc$ of rank $r$.

\begin{lemma}
Let $\Cc$ be a convolutional code. Let $\mathcal{V}$ denote the set of $\F_q$-linear subspaces of $\Cc$. Then
\begin{enumerate}
\item $d_r(\Cc)=\min\{|\supp(U)|: U\subseteq\Cc\text{ and } \rk(\langle U\rangle_{\F_q[x]})\geq r\}$,
\item $d_r(\Cc)=\min\{|\supp(U)|: U\subseteq\Cc\text{ and }\rk(\langle U\rangle_{\F_q[x]})= r\}$,
\item $d_r(\Cc)=\min\{|\supp(U)|: U\in\mathcal{V}\text{ and } \rk(\langle U\rangle_{\F_q[x]})\geq r\}$,
\item $d_r(\Cc)=\min\{|\supp(U)|: U\in\mathcal{V}\text{ and }\rk(\langle U\rangle_{\F_q[x]})= r\}$,
\item $d_r(\Cc)=\min\{\wt(\mathcal{D})\mid \mathcal{D}\subseteq\Cc\mbox{ is a subcode of } \rk(\mathcal{D})=r\}$,
\end{enumerate}
for $1\leq r\leq k$.
\end{lemma}

\begin{proof}
We already observed that the support of a linear space is equal to the union of the supports of a system of generators. Therefore, for any set $U$, the support of $U$ is the same as the support of the $\F_q$-linear space generated by $U$. This proves that definitions 1. and 3. are equivalent and 2. and 4. are equivalent. Moreover, by comparing the sets over which we minimize, one sees that the minimum in 4. is greater than or equal to that in 3., the minimum in 5. is greater than or equal to the $r$-th generalized weight of $\Cc$, the minimum in 5. is greater than or equal to that in 4., and the $r$-th generalized weight of $\Cc$ is greater than or equal to the minimum in 3. In order to prove that all numbers coincide, it suffices to show that the minimum in 3. is greater than or equal to that in 5.

Let $U\subseteq \Cc$ be an $\F_q$-linear subspace such that $\rk(\langle U\rangle_{\F_q[x]})\geq r$. Then $\dim(U)\geq r$ and there exists $U'\subseteq U$ an $\F_q$-linear subspace such that $\dim(U')=\rk(\langle U'\rangle_{\F_q[x]})=r$. 
We conclude, since $U'\subseteq U$ implies that $\supp(U')\subseteq \supp(U)$.
\end{proof}
	
\begin{remark}\label{rmk:noncatsubcodes}
Consider a noncatastrophic code $\Cc$. One may also define the generalized weights as
\begin{equation*}
\begin{split}	
\tilde d_r(\Cc)=&\min\{\wt(\mathcal{D})\mid \mathcal{D}\subseteq\Cc\text{ is a noncatastrophic subcode of }\rk(\mathcal{D})=r\}
\end{split}
\end{equation*}
for $1\leq r\leq k$. Notice that this definition is not equivalent to Definition~\ref{defgenw}. Indeed, it may happen that $\tilde d_1(\Cc)\neq \df(\Cc)$, while $d_1(\Cc)=\df(\Cc)$ for every $\Cc$, as shown in Proposition~\ref{propproperties}. For example, let $\Cc=\langle (1,1+x+x^2+x^3)\rangle_{\F_2[x]}\subseteq\F_2[x]^2$. Then $\df(\Cc)=\wt(1+x,1+x^4)=4$ and every element of minimum weight generates a catastrophic code. On the other side, $\wt(1,1+x+x^2+x^3)=5$ hence $\tilde d_1(\Cc)=5$. Since we want the first generalized weight to be equal to $\df(\Cc)$, we will not discuss this definition further.
\end{remark}
	
In the next proposition, we establish some basic properties of the generalized weights. In particular, we prove that they are strictly increasing and that the minimum distance coincides with the first generalized weight. Moreover, we provide an upper bound on each generalized weight. 

\begin{proposition}\label{propproperties}
Let $\Cc\subseteq \mathcal{D}\subseteq \F_q[x]^n$ be convolutional codes, let $k=\rk(\Cc)$. Then
\begin{enumerate}
\item $d_1(\Cc)=\df(\Cc)$.
\item $d_r(\Cc)<d_{r+1}(\Cc)$ for all $1\leq r\leq k-1$.
\item $d_r(\mathcal{D})\leq d_{r}(\Cc)$ for all $1\leq r\leq k$.
\item $d_r(\Cc)\leq n(\delta_1+1)-k+r$ for all $1\leq r\leq k$.
\item $d_k(\Cc)\leq\wt(\Cc)$.
\end{enumerate}
\end{proposition}

\begin{proof}
1., 3., and 5. follow directly from the definition, while 4. follows by combining 2. and 5.
		
For 2., let $\mathcal{U}=\langle u_1,\dots,u_{r+1}\rangle_{\F_q[x]}$ be a subcode of $\Cc$ that realizes $d_{r+1}(\Cc)$. After adding suitable multiples of $u_1$ to the other generators, we may suppose that $\supp(u_1)\not\subseteq\supp(\langle u_2,\dots,u_{r+1}\rangle_{\F_q})$. Since $u_2,\dots,u_{r+1}$ are still $\F_q[x]$-linearly independent, we conclude.
\end{proof}

\begin{remark}
Notice that, unlike what happens for linear block codes, one does not always have that $d_k(\Cc)=\wt(\Cc)$. For example, for the code $\Cc$ of Remark~\ref{rmk:noncatsubcodes}, one has $d_1(\Cc)=4<5=\wt(\Cc)$.
\end{remark}

\begin{remark}
In spite of their simplicity, for all these bounds there exists codes that meet them. In particular, in Proposition~\ref{prop:MDSbounds} we prove that MDS codes meet bound 4. and in Proposition~\ref{proposition:weightsoptimal1} we prove that the generalized weights of optimal anticodes increase by one at each step.
\end{remark}

Generalized weights are invariant under isometries of convolutional codes. As a consequence, they are also invariant under strong isometries as defined in~\cite{Glu}.

\begin{proposition}\label{proposition:isometric}
If $\Cc_1$ and $\Cc_2$ are isometric convolutional codes, then they have the same weight and generalized weights.
\end{proposition}

\begin{proof}
Let $\phi:\Cc_1\rightarrow\Cc_2$ be an isometry. Since $\phi^{-1}:\Cc_2\rightarrow\Cc_1$ is also an isometry, it suffices to prove that the weight and the generalized weights of $\Cc_1$ are greater than or equal to the corresponding invariants of $\Cc_2$.

Let $c_1(x),\ldots,c_k(x)$ be a basis of $\Cc_1$ such that $\wt(\Cc_1)=|\supp(\langle c_1(x),\ldots,c_k(x)\rangle_{\F_q})|$. Then $\phi(c_1(x)),\ldots,\phi(c_k(x))$ are a basis of $\Cc_2$ and the restriction of $\phi$ is an $\F_q$-linear isometry of linear block codes between $\langle c_1(x),\ldots,c_k(x)\rangle_{\F_q}$ and $\langle \phi(c_1(x)),\ldots,\phi(c_k(x))\rangle_{\F_q}$, with respect to the Hamming distance. In particular,
$$\wt(\Cc_1)=|\supp(\langle c_1(x),\ldots,c_k(x)\rangle_{\F_q})|=|\supp(\langle \phi(c_1(x)),\ldots,\phi(c_k(x))\rangle_{\F_q})|\geq\wt(\Cc_2).$$ 
Suppose now that $\mathcal{D}\subseteq\Cc_1$ realizes $d_r(\Cc_1)$. Since $\phi$ is an isomorphism of $\F_q[x]$-modules, then $\rk(\phi(\mathcal{D}))=r$. Moreover, $\phi$ induces an isometry between $\mathcal{D}$ and $\phi(\mathcal{D})$, hence
\begin{equation*}
    d_r(\Cc_1)=\wt(\mathcal{D})=\wt(\phi(\mathcal{D}))\geq d_r(\Cc_2).\qedhere
\end{equation*}
\end{proof}

In~\cite{RY}, the authors introduce the family of generalized Hamming weights for convolutional codes. We now recall their definition and we briefly discuss how it relates to the generalized weights that we introduced.

\begin{definition}\label{definition:genwei2}
Let $\Cc$ be a convolutional code. For every positive integer $r$, the r-th \textbf{generalized Hamming weight} of $\Cc$ is
\begin{equation*}
d'_r(\Cc)=\min\{|\supp(U)|:U \text{ is an }\F_q\text{-linear subspace of }\Cc\text{ and }\dim(U)=r\}.
\end{equation*}
\end{definition}

\begin{remark}
It follows directly from the definitions that $d_i'(\Cc)\leq d_i(\Cc)$ for $1\leq i\leq \rk(\Cc)$.
\end{remark}

From now on, we refer to the weights from Definition~\ref{defgenw} as generalized weights and to those from Definition~\ref{definition:genwei2} as generalized Hamming weights. Even though Definition~\ref{definition:genwei2} appears to be similar to our Definition~\ref{defgenw}, the fact that we consider rank $r$ subcodes in place of $r$-dimensional subspaces leads to a different set of invariants. 

In the next examples we exhibit two pairs of non-isometric codes. The codes in the first example can be distinguished using the generalized weights, but not using the generalized Hamming weights, while the codes in the second example can be distinguished using the generalized Hamming weights, but not using the generalized weights. 

\begin{example}
(a) Let $\Cc_1, \Cc_2\in\F_q[x]^3$ be the convolutional codes generated respectively by $(1,0,0),(0,1,1+x)$ and $(1,0,0),(0,0,1)$. Then, $d'_r(\Cc_1)=d'_r(\Cc_2)=r$ for every positive integer $r$. By computing the generalized weights according to Definition~\ref{defgenw}, we can prove that the two codes are not isometric, since $d_1(\Cc_1)=d_1(\Cc_2)=1$, $d_2(\Cc_1)=4$ and $d_2(\Cc_2)=2$.

(b) Let $\Cc_1,\Cc_2\in\F_q[x]^3$ be the convolutional codes generated respectively by $(1,1,1)$ and $(1+x,0,1)$. Then $d_1(\Cc_1)=d_1(\Cc_2)$, that is, the generalized weights of $\Cc_1$ and $\Cc_2$ coincide. However, $d'_2(\Cc_1)=6$ and $d'_2(\Cc_2)=5$, in particular the codes are not isometric, as generalized Hamming weights are invariant under isometry. 
\end{example}
	
Notice that an $(n,k,\delta)$ convolutional code has exactly $k$ generalized weights and an infinite number of generalized Hamming weights. In particular, one can recover the rank of a convolutional code from its generalized weights, but not from its generalized Hamming weights, as the next example shows.

\begin{example}
Let $\Cc\subseteq\F_2[x]^n$ be such that $(1,0,\ldots,0)\in\Cc$. Then $d'_r(\Cc)=r$ for any $r\geq 1$. Clearly, there exist codes of any rank $k\leq n$ which contain the codeword $(1,0,\ldots,0)$.
\end{example}

It may happen that, in order to distinguish two non-isometric codes, one needs to compute an arbitrarily large number of generalized Hamming weights. For instance, in the next example we show that, for fixed $n$ and $k$, there exist non-isometric convolutional codes with the same first $N$ generalized weights for $N$ arbitrary large as $\delta$ goes to infinity. We will use the following simple lemma.

\begin{lemma}\label{lemma:sumomega}
Let $q$ be a prime number. For a polynomial $p(x)\in\F_q[x]\setminus\{0\}$ we have that 
\begin{equation*}
\wt\left(p(x)\sum_{t=0}^{N}x^{q^t}\right)+\wt(p(x))\geq N+2.
\end{equation*}
\end{lemma}

\begin{proof}
Let $p(x)=a_1x^{k_1}+\dots+a_{\ell}x^{k_{\ell}}$ be a polynomial with $k_1<\dots<k_{\ell}$ and $\wt(p(x))=\ell>0$. If $\ell\geq N+1$ or $\ell=1$ the statement is trivially true. Suppose $\ell< N+1$. Clearly,
\begin{equation}\label{eqsumomega2}
p(x)\sum_{t=0}^{N}x^{q^t}=\sum_{t=0}^{N}\sum_{i=1}^{\ell}a_ix^{k_i+q^t}.
\end{equation}
Two monomials have the same exponent if and only if there are $i_1,i_2$, $t_1,t_2$ such that $k_{i_1}+q^{t_1}=k_{i_2}+q^{t_2}$. Moreover if $(t_1,t_2)\neq(t_3,t_4)$ with $t_1<t_2$ and $t_3<t_4$ then $q^{t_2}-q^{t_1}\neq q^{t_4}-q^{t_3}$. Therefore, in \eqref{eqsumomega2} there are at most $\ell(\ell-1)/2$ pairs of monomials with the same exponent. Since the number of monomials in the sum is $(N+1)\ell$, we have that
\begin{equation*}
\wt\left(p(x)\sum_{t=0}^{N}x^{q^t}\right)+\ell\geq(N+1)\ell-\ell(\ell-1)+\ell=\ell(N-\ell+3).
\end{equation*}
Finally, since $1<\ell<N+1$, we have that $\ell(N-\ell+3)\geq N+2$.
\end{proof}
	
\begin{example}
Let $q$ be a prime number and let $\Cc_N=\langle(1,1,0),(0,\sum_{t=0}^{N}x^{q^t},1)\rangle_{\F_q[x]}\subseteq\F_q[x]^3$.
We claim that $d'_r(\Cc_N)=2r$ for all $r\in\{1,\dots,N+1\}$ and $d'_{N+2}(\Cc_N)=2(N+1)+1$. In particular, the first $N$ generalized weights of $\Cc_{N-1}$ and $\Cc_{N}$ coincide and $d'_{N+1}(\Cc_{N-1})\neq d'_{N+1}(\Cc_{N})$.
	
Let $U=\langle p_{i,1}(x)(1,1,0)+p_{i,2}(x)(0,\sum_{t=0}^N x^{q^t},1)\mid 1\leq i\leq r\rangle_{\F_q}\Cc_N$ be a linear subspace of dimension $r\leq N+1$. Let $J\subseteq\{1,\dots,r\}$ be a maximal set of indices such that $\{p_{j,1}(x)\}_{j\in J}$ is an $\F_q$-linear independent set. If $|J|=r$, then $d'_r(\Cc_N)\geq 2r$.
If $|J|<r$, we may assume without loss of generality that $p_{i,1}(x)=0$ for every $i\notin J$. It is easy to show that the support of $U$ has cardinality at least $|J|$, when restricted to the first component. The set $\{p_{j,2}(x)\}_{j\notin J}$ is $\F_q$-linearly independent by assumption. Moreover, after replacing the elements of the set with appropriate linear combinations, we may assume that there exists a $\bar j\notin J$ such that the last entry of the support of $U$ has cardinality greater than or equal to $r-|J|-1+\wt(p_{\bar j,2}(x))$. Since the cardinality of the second entry of the support of $U$ is at least $\wt\left(p_{\bar j,2}(x)\sum_{t=0}^{N}x^{q^t}\right)$, then
$$d'_r(\Cc_N)\geq |J|+r-|J|-1+\wt(p_{\bar j,2}(x))+\wt\left(p_{\bar j,2}(x)\sum_{t=0}^{N}x^{q^t}\right).$$
By Lemma~\ref{lemma:sumomega}, we conclude that $$d'_r(\Cc_N)\geq N+1+r\geq2r.$$
Let $V=\langle (x^i,x^i,0)\mid 1\leq i\leq r\rangle_{\F_q}$. Then $V$ is an $r$-dimensional subspace of $\Cc_N$ with $|\supp(V)|=2r$, showing that $d'_r(\Cc_N)=2r$ for $1\leq r\leq N+1$. 

Moreover, $V=\langle (0,\sum_{t=0}^{N}x^{q^t},1), (x^{q^i},x^{q^i},0)\mid \leq i\leq N\rangle_{\F_q}$ is an $(N+2)$-dimensional $\F_q$-linear subspace of $\Cc$, showing that $d'_{N+2}(\Cc_N)\leq 2(N+1)+1$. Since $d'_{N+1}=2(N+1)$, then $d'_{N+2}(\Cc_N)=2(N+1)+1$ by Proposition~\ref{propproperties}.
\end{example}

The next proposition relates the generalized weights of an $(n,k,0)$ convolutional code with the generalized Hamming weights of the linear block code generated by the same elements. The latter were studied by Wei in~\cite{Wei} and we refer to his work for their definition and basic properties. 
	
\begin{proposition}\label{proposition:hammingweights}
Let $\Cc$ be an $(n,k,0)$ convolutional code. Then $\rk(\Cc)=\dim(\Cc[0])$ and
$$d_r(\Cc)=d_r^H(\Cc[0])$$
for $1\leq r\leq\rk(\Cc)$, where $d_r^H(\Cc[0])$ denotes the $r$-th generalized Hamming weight of $\Cc[0]$. In particular
$$d_r(\Cc)\leq n+k-r.$$
In addition, if $D\subseteq \Cc[0]$ is an $r$-dimensional subspace such that $d_r^H(\Cc[0])=|\supp(D)|$, then $\mathcal{D}=\langle D\rangle_{\F_q[x]}\subseteq\Cc$ realizes $d_r(\Cc)$. In particular $$d_{\rk(\Cc)}(\Cc)=\wt(\Cc)=|\supp(\Cc[0])|,$$ where $\supp(\Cc[0])$ denotes the Hamming support of $\Cc[0]$.
\end{proposition}
	
\begin{proof}
Since $\Cc=\langle\Cc[0]\rangle_{\F_q[x]}$ and $\Cc[0]\subseteq\F_q^n$, then $\rk(\Cc)=\dim(\Cc[0])$. 
Fix $1\leq r\leq \rk(\Cc)$ and let $D\subseteq \Cc[0]$ be an $\F_q$-linear subspace such that $\dim(D)=r$ and $d_r^H(\Cc[0])=|\supp(D)|$. Let $\mathcal{D}=\langle D\rangle_{\F_q[x]}\subseteq\Cc$, then $\rk(\mathcal{D})=r$ and $\wt(\mathcal{D})\leq |\supp(D)|=d_r^H(\Cc[0])$. This implies that $$d_r^H(\Cc[0])\geq d_r(\Cc).$$

To prove the reverse inequality, let $\mathcal{D}\subseteq\Cc$ be a rank $r$ subcode and let $v_1,\dots,v_r$ be an $\F_q[x]$-basis of $\mathcal{D}$ such that $\wt(\mathcal{D})=|\supp(\langle v_1,\ldots,v_r\rangle_{\F_q})|$. For $1\leq i\leq r$ write $v_i=\sum_{j=0}^{t_i}v_i[j]x^{j}$. Since $\Cc=\langle\Cc[0]\rangle_{\F_q[x]}$, then $v_i[j]\in\Cc[0]$ for all $i,j$. Let $U=\langle\{u_{j,i}\}_{j,i}\rangle_{\F_q}$. By construction $$\mathcal{D}\subseteq \langle U\rangle_{\F_q[x]}\text{  and  }\wt(\mathcal{D})=|\supp(\langle v_1,\ldots,v_r\rangle_{\F_q})|\geq|\supp(U)|.$$ Since $\dim(U)=\rk(\langle U\rangle_{\F_q[x]})\geq\rk(\mathcal{D})=r$, one can find $U'\subseteq U$ an $\F_q$-linear subspace with $\dim(U')=r$. Then $\rk(\langle U'\rangle_{\F_q[x]})=r$ and $$\wt(\langle U'\rangle_{\F_q[x]})\leq |\supp(U')|\leq|\supp(U)|\leq\wt(\mathcal{D}).$$ 
Observe in addition that, if $\mathcal{D}$ has an $\F_q[x]$-basis $v'_1,\ldots,v'_r$ which consists of elements of $\F_q^n$, then $U=U'=\langle v'_1,\ldots,v'_r\rangle_{\F_q}$, showing that $$\wt(\mathcal{D})=|\supp(\langle v'_1,\ldots,v'_r\rangle_{\F_q})|.$$
Summarizing we have shown that, for every submodule $\mathcal{D}\subseteq\Cc$ of rank $r$, one can find a submodule of $\Cc$ of the form $\langle U'\rangle_{\F_q[x]}$ for some $U'\subseteq \Cc[0]$, with rank $r$ and weight smaller than or equal to the weight of $\mathcal{D}$. Since $|\supp(U')|=\wt(\langle U'\rangle_{\F_q[x]})$, this implies that $d_r^H(\Cc[0])\leq d_r(\Cc)$.
\end{proof} 
	
Notice that the generalized Hamming weights of an $(n,k,0)$ convolutional code may not coincide with the generalized Hamming weights of the linear block code generated by the same elements, as the next example shows.

\begin{example}
Let $\Cc=\langle (1,0,0),(0,1,1)\rangle_{\F_2[x]}\subseteq\F_2^3$ and let $\Cc[0]\subseteq\F_2^3$. The generalized Hamming weights of $\Cc[0]$ are $d_1^H(\Cc[0])=1$ and $d_2^H(\Cc[0])=3$, while the generalized Hamming weights of $\Cc$ are $d'_r(\Cc)=r$ for all $r\geq 1$. 
\end{example}
	
The next result is obtained by combining Proposition~\ref{proposition:hammingweights} and Wei duality for the generalized Hamming weight of a linear block code.

\begin{proposition}
Let $\Cc$ be an $(n,k,0)$ convolutional code. Then $\Cc^{\perp}$ is an $(n,n-k,0)$ convolutional code and its set of generalized weights is $$\{d_r(\Cc^{\perp})\mid 1\leq r\leq n-k\}=\{n+1-d_r(\Cc)\mid 1\leq r\leq k\}.$$ In particular, the generalized weights of $\Cc$ determine those of $\Cc^{\perp}$.
\end{proposition}

\begin{proof}
By Proposition~\ref{prop:C[0]}, $\Cc$ has the form $\Cc=\langle \Cc[0]\rangle_{\F_q[x]}$ and $\dim(\Cc[0])=k$. Moreover, $\Cc^{\perp}=\langle \Cc[0]^{\perp}\rangle_{\F_q[x]}$ is an $(n,n-k,0)$ convolutional code. Wei duality states that $$\{d_r^H(\Cc[0]^{\perp})\mid 1\leq r\leq n-k\}=\{n+1-d_r^H(\Cc[0])\mid 1\leq r\leq k\},$$ see~\cite[Theorem 3]{Wei}. We conclude by Proposition~\ref{proposition:hammingweights}.
\end{proof}

For a catastrophic convolutional code, one may have that $\Cc\subsetneq(\Cc^{\perp})^{\perp}$. In such a situation, one expects to be able to find a code such that $\Cc$ and $(\Cc^{\perp})^{\perp}$ have different generalized weights. This is indeed the case, as the next example shows.

\begin{example}
Let $\Cc=\langle (1+x,0)\rangle_{\F_2[x]}\subseteq\F_2[x]^2$. Then $\Cc^{\perp}=\langle(0,1)\rangle$ and $(\Cc^{\perp})^{\perp}=\langle(1,0)\rangle\supsetneq\Cc$. In addition, $d_1(\Cc)=2$ while $d_1((\Cc^{\perp})^{\perp})=1$. In particular, $\Cc$ and $(\Cc^{\perp})^{\perp}$ have different generalized weights, while having the same dual code $\Cc^{\perp}$.
\end{example}

Noncatastrophic convolutional codes coincide with their double dual. However, no result along the lines of Wei duality holds even when restricting to this class of codes.
More precisely, the next example shows that there exists noncatastrophic convolutional codes with the same generalized weights and whose dual codes have different generalized weights.

\begin{example}
Let $\Cc_1=\langle (1+x,1+x,1,0)\rangle_{\F_q[x]}$ and $\Cc_2=\langle(1+x,1,1,1)\rangle_{\F_q[x]}$ be $(4,1,1)$ noncatastrophic convolutional codes. We have that $d_1(\Cc_1)=d_1(\Cc_2)=5$, while $d_1(\Cc_1^{\perp})=1$ and $d_1(\Cc_2^{\perp})=2$.
\end{example}

Notice that the previous example also shows that Wei duality cannot hold for any set of generalized weights with the property that the first generalized weight is the free distance of the code. More precisely, we have the following.

\begin{remark}
Given any definition of generalized weights for convolutional codes such that the first generalized weight is the free distance of the code, the generalized weights of a code do not determine in general the generalized weights of its dual.
\end{remark}

Although Wei duality does not hold for this type of dual, the are other types of duality that have been considered in the literature.  For example, in~\cite{RST} the authors define the reverse of a convolutional code. 

\begin{definition}\label{defn:reverse}
Let $\Cc$ be an $(n,k,\delta)$ convolutional code and let $\rev:\F_q[x]^n\rightarrow \F_q[x]^n$ be the map given by $\rev(0)=0$ and
$$\rev(c(x))=x^{\deg(c(x))}c\left(\frac{1}{x}\right)$$
if $c(x)\neq 0$. The \textbf{reverse code} of $\Cc$ is $$\rev(\Cc)=\langle\rev(c(x)):c(x)\in\Cc\rangle_{\F_q[x]}.$$
\end{definition}

\begin{remark}
Let $c(x)\in\F_q[x]^n\setminus\{0\}$. The following are equivalent:
\begin{itemize}
\item $\rev(\rev(c(x)))=c(x)$,
\item $\deg(c(x))=\deg(\rev(c(x)))$,
\item $x\nmid c(x)$.
\end{itemize}
In addition, one has $x^d\rev(\rev(c(x)))=c(x)$ for $d=\max\{t\geq 0 : x^t\mid c(x)\}$. 
\end{remark}

For the sake of completeness, we prove that Definition~\ref{defn:reverse} is equivalent to the definition of reverse code given in~\cite{RST}. This implies in particular that if $\Cc$ is an $(n,k,\delta)$ convolutional code, then $\rev(\Cc)$ is an  $(n,k,\delta')$ convolutional code, for some $\delta'$.
	
\begin{proposition}\label{proposition:revcode}
Let $\Cc$ be an $(n,k,\delta)$ convolutional code and let $G$ be a row-reduced generator matrix for $\Cc$. Let $c_1,\dots, c_k$ be the rows of $G$. Then,
\begin{equation*}
\rev(\Cc)=\langle \rev(c_1), \dots, \rev(c_k)\rangle_{\F_q[x]}.
\end{equation*}
In particular, $\rk(\rev(\Cc))=k$.
\end{proposition}

\begin{proof}
It is clear from the definition that $\rev(\Cc)\supseteq\langle \rev(c_1), \dots, \rev(c_k)\rangle_{\F_q[x]}$. Therefore, it suffices to show that $\rev(c)\in\langle \rev(c_1), \dots, \rev(c_k)\rangle_{\F_q[x]}$ for all $c\in\Cc$. Let $c\in\Cc$ and let $u_1,\dots,u_k\in\F_q[x]$ such that $c=\sum u_ic_i$. Since $G$ is row-reduced, we have that $\deg(c)=\max_i\{deg(u_i)+\deg(c_i)\}$. Therefore
\begin{equation*}
\rev(c)=x^{\deg(c)}\sum_{i=1}^{k} u_i\left(\frac{1}{x}\right)c_i\left(\frac{1}{x}\right)=\sum_{i=1}^{k} x^{\deg(c)-\deg(c_i)}u_i\left(\frac{1}{x}\right)\rev(c_i).
\end{equation*}
This concludes the proof. Indeed, from $\deg(c)=\max_i\{deg(u_i)+\deg(c_i)\}$ we immediately deduce that $x^{\deg(c)-\deg(c_i)}u_i\left(\frac{1}{x}\right)\in\F_q[x]$. 
\end{proof}

\begin{corollary}\label{corollary:rev}
Let $\Cc$ be an $(n,k,\delta)$ convolutional code. The following hold:
\begin{enumerate}
\item $\rev(\rev(\Cc))\supseteq\Cc$.
\item There exists a positive integer $d$, such that $x^d\rev(\rev(\Cc))\subseteq\Cc$.
\item If $\Cc$ is noncatastrophic, then $\rev(\rev(\Cc))=\Cc$.
\end{enumerate}
\end{corollary}
	
\begin{proof}
1. Let $G$ be a row-reduced generator matrix for $\Cc$ and let $c_1,\dots, c_k$ be the rows of $G$. For every $1\leq i\leq k$ there exists $t_i\geq 0$ such that $x^{t_i}\rev(\rev(c_i))=c_i$. Therefore, $\rev(\rev(\Cc))\supseteq\Cc$. 

2. Consider the ascending chain of modules 
\begin{equation}\label{eqn:chain}
\Cc\subseteq\Cc:x\subseteq\ldots\subseteq\Cc:x^d\subseteq\ldots
\end{equation}
where $\Cc:x^d=\{c(x)\in\F_q[x]^n\mid x^dc(x)\in\Cc\}$ for $d\geq 0$. Since every submodule of $\F_q[x]^n$ is finitely generated, any ascending chain of submodules of $\F_q[x]^n$ is stationary. This means that there exists a $d$ such that $\Cc:x^d=\Cc:x^e$ for any $e\geq d$. Let $c\in\Cc$, then $c=x^t\rev(\rev(c))$ for some $t\geq 0$, hence $\rev(\rev(c))\in\Cc:x^t\subseteq\Cc:x^d$ if $d\geq t$. Moreover, $\Cc:x^d=\Cc:x^t$ if $t\geq d$. Therefore $\Cc:x^d\supseteq\rev(\rev(\Cc))$, that is $$x^d\rev(\rev(\Cc))\subseteq\Cc$$ for $d\gg 0$.

3. For a noncatastrophic code one has $\{c(x)\in\F_q[x]^n\mid x^dc(x)\in\Cc\}=\Cc$ for all $d\geq 0$ by Proposition~\ref{proposition:noncatastrophic}. Therefore, all the containments in (\ref{eqn:chain}) are equalities. In particular $$\rev(\rev(\Cc))\subseteq\Cc:x^d=\Cc\subseteq\rev(\rev(\Cc))$$
where the first containment follows from 2. and the second from 1.
\end{proof}

The next proposition proves that the generalized weights of $\Cc$ are the same as those of $\rev(\Cc)$.

\begin{proposition}\label{proposition:reverse}
Let $\Cc$ be an $(n,k,\delta)$ convolutional code. Then,
$$d_r(\Cc)=d_r(\rev(\Cc)),$$
for $1\leq r\leq k$.
\end{proposition}

\begin{proof}
Let $\mathcal{D}=\langle c_1,\dots ,c_r\rangle_{\F_q[x]}$ be a submodule of $\Cc$ that realizes $d_r(\Cc)$ and such that $\wt(\mathcal{D})=|\supp(\langle c_1,\dots ,c_r\rangle_{\F_q})|$. Let $s=\max\{\deg(c_i):1\leq i\leq r\}$ and let
$$\mathcal{D}'=\langle x^{s-\deg(c_1)}\rev(c_1),\dots,x^{s-\deg(c_r)}\rev(c_r)\rangle_{\F_q[x]}\subseteq \rev(\Cc).$$ Since $\rk\left(\mathcal{D}'\right)=\rk\left(\rev(\mathcal{D})\right)=r$, then $$d_r(\rev(\Cc))\leq\wt(\mathcal{D}')\leq|\supp(\langle x^{s-\deg(c_1)}\rev(c_1),\dots,x^{s-\deg(c_r)}\rev(c_r)\rangle_{\F_q})|=\wt(\mathcal{D})=d_r(\Cc).$$ To prove the reverse inequality, let $\mathcal{D}\subseteq\rev(\rev(\Cc))$ be a submodule that realizes $d_r(\rev(\rev(\Cc)))$. By Corollary~\ref{corollary:rev} there exists a positive integer $d$ such that $x^{d}\mathcal{D}\subseteq \Cc$. Since $\rk(x^{d}\mathcal{D})=\rk(\mathcal{D})=r$, then $$d_r(\Cc)\leq\wt(x^d\mathcal{D})=\wt(\mathcal{D})=d_r(\rev(\rev(\Cc))).$$
Therefore
\begin{equation*}
d_r(\Cc)\leq d_r(\rev(\rev(\Cc)))\leq d_r(\rev(\Cc))\leq d_r(\Cc).\qedhere
\end{equation*}
\end{proof}

\section{Minimal supports}

In this section, we study codewords of minimal support and submodules that realize the generalized weights of a convolutional code. We show that it is possible to calculate the generalized weights considering only subspaces with special properties. In particular, in Theorem~\ref{theorem:bound} we prove that, in order to compute $d_r(\Cc)$, we may restrict to subspaces generated by vectors, whose degree is bounded by a function of $n,k,r$ and $\delta_1$ only. Moreover, we show that the generalized weights are realized by subspaces generated by codewords of minimal support. 

Some of the results in this section are similar to those obtained in~\cite[Section 3]{GR22} for a large family of support functions and codes over rings. Notice however that the setup of~\cite{GR22} does not apply to our situation, as the Hamming support for convolutional codes is not a support according to the definition from~\cite{GR22}. 

\begin{definition}
Let $\Cc\subseteq\F_q[x]^n$ be a convolutional code. A codeword $c\in\Cc$ is {\bf minimal} if its support is minimal among the supports of the nonzero codewords of $\Cc.$
\end{definition}	

It is easy to show that, for a given code, minimal supports correspond uniquely to minimal codewords, up to a nonzero scalar multiple.

\begin{lemma}
Let $\Cc\subseteq\F_q[x]^n$ be a convolutional code. If two minimal codewords $u(x),v(x)\in\Cc$ have the same support, then there exists $\alpha\in\F_q^*$ such that $u(x)=\alpha v(x)$.  
\end{lemma}

\begin{proof}
If $u(x),v(x)\in\Cc$ have the same support, then there exists $\alpha\in\F_q^*$ such that $\supp(u(x)-\alpha v(x))\subsetneq\supp(u(x))$. By the minimality of the support of $u(x)$, we deduce that $u(x)-\alpha v(x)=0$.
\end{proof}

In this section, we study the subcodes of $\Cc$ which realize its generalized weights. We start by showing that if $\Cc$ is a noncatastrophic convolutional code, then each of its generalized weights is realized by a subspace that contains an element that is not divisible by $x$.

\begin{theorem}
Let $\Cc$ be a noncatastrophic convolutional code of rank $k$. For $1\leq r\leq k$, consider the set
\begin{equation*}
\mathcal{U}_r=\{\mathcal{D}\subseteq\Cc\text{ is a subcode of }\rk(\mathcal{D})=r \mbox{ and }\exists\; c\in \mathcal{D}\text{ with }c[0]\neq 0\}.
\end{equation*}
Then
\begin{equation*}
d_r(\Cc)=\min\{\wt(\mathcal{D})\mid \mathcal{D}\in\mathcal{U}_r\}.
\end{equation*}
\end{theorem}
	
\begin{proof}
Let $\mathcal{D}=\langle c_1,\dots, c_r\rangle_{\F_q[x]}\subseteq\Cc$ be a subcode that realizes the $r$-th generalized weight and such that $\wt(\mathcal{D})=|\supp(\langle c_1,\dots, c_r\rangle_{\F_q})|$. Let $\ell=\max\{d\geq 0 : x^d\mid c_i \mbox{ for } 1\leq i\leq k\}$. Since $\Cc$ is noncatastrophic, by Proposition~\ref{proposition:noncatastrophic} we have that $\mathcal{D}'=\langle c_1/x^{\ell},\dots, c_r/x^{\ell}\rangle\in\mathcal{U}_r$. This concludes the proof, since $\wt(\mathcal{D})=|\supp(\langle c_1,\dots, c_r\rangle_{\F_q})|=|\supp(\langle c_1/x^{\ell},\dots, c_r/x^{\ell}\rangle_{\F_q})|\geq\wt(\mathcal{D}')$.
\end{proof}
	
Let $\phi_{d,\delta}:\F_q[x]\rightarrow \F_q[x]_{\leq\delta}$ be the linear map given by $$\phi_{d,\delta}(a_0x+\dots+a_{d-\delta}x^{d-\delta}+\dots a_dx^{d}+\dots+a_nx^n)=a_{d-\delta}+a_{d-\delta+1}x+\dots a_dx^{\delta},$$
where $\F_q[x]_{\leq\delta}$ denotes the set of polynomials of degree at most $\delta$. 
We can extend this map to $\F_q[x]^{k\times n}$ by applying $\phi_{d,\delta}$ to each entry. This yields the map $\Phi_{d,\delta}:\F_q[x]^{k\times n}\rightarrow \F_q[x]^{k\times n}_{\leq\delta}$, given by 
$$\Phi_{d,\delta}((m_{i,j}(x))_{i,j})=(\phi_{d,\delta}(m_{i,j}(x))))_{i,j}.$$
In order to simplify the notation, given a matrix $M\in \F_q[x]^{k\times n}$, we write $M_t=\Phi_{t,t}(M)$.
The next lemma follows directly from the definition.

\begin{lemma}\label{lemma:translation}
Let $M\in\F_q[x]^{k\times n}$ be a matrix with entries in $\F_q[x]$.
\begin{enumerate}
\item If each entry of $M$  is divisible by $x^t$ for some $t\in\mathbb{N}$, then
\begin{equation*}
\Phi_{d,\delta}(x^{-t}M)=\Phi_{d+t,\delta}(M).
\end{equation*}
\item If $d-\delta\leq\deg(M)\leq d$, then
\begin{equation*}
\Phi_{d,\delta}(M)=\Phi_{\deg(M),\delta-d+\deg(M)}(M).
\end{equation*}
\end{enumerate}
\end{lemma}

The next lemma is crucial for the proof of Theorem~\ref{theorem:bound}.

\begin{lemma}\label{lemma:reduction}
Let $\Cc$ be an $(n,k,\delta)$ convolutional code, $G$ a row-reduced generator matrix for $\Cc$ with maximum degree $\delta_1\geq 1$, $M=(m_{i,j}(x))_{i,j}\in\F_q[x]^{r\times k}$, and $1\leq s_1<s_2\leq d=\max_{i,j} \deg(m_{i,j}(x))$ natural numbers. If $s_2-s_1\geq q^{\delta_1kr}$, then there exist $M'=(m'_{i,j}(x))_{i,j}\in\F_q[x]^{r\times k}$  and a natural number $t<d'=\max_{i,j} \deg(m'_{i,j}(x))$ such that
\begin{enumerate}
\item $s_1\leq t<t+d-d'\leq s_2$.
\item $(M'G)_t=(MG)_t$.
\item $\Phi_{d+\delta_1,d'+\delta_1-t-1}(MG)=\Phi_{d'+\delta_1,d'+\delta_1-t-1}(M'G)$.
\end{enumerate}
\end{lemma}

\begin{proof}
Consider the set $F=\{\Phi_{s+\delta_1,\delta_1-1}(M_sG):s_1\leq s\leq s_2\}$. Since $s+\delta_1>\delta_1$ and the rows of $M_sG$ are elements of $\Cc$ of degree smaller than or equal to $s+\delta_1$, then
\begin{equation*}
|F|\leq \left|\bigcup_{h=\delta_1}^{\infty}\left\{\Phi_{h,\delta_1-1}\begin{pmatrix}
c_1\\\vdots\\c_r
\end{pmatrix}:c_i\in\Cc\text{ and }\deg(c_i)\leq h\right\}\right|\leq q^{\delta_1kr}.
\end{equation*}
The second inequality follows from observing that there are $r$ rows, each row is a combination of $k$ generators of $\Cc$ and there are $\delta_1$ possible shifts.

Since $s_2-s_1\geq q^{\delta_1kr}$, by the pigeonhole principle there exist $s_1\leq t<t'\leq s_2$ such that $\Phi_{t+\delta_1,\delta_1-1}(M_{t}G)=\Phi_{t'+\delta_1,\delta_1-1}(M_{t'}G)$. Let
\begin{equation}\label{equation:M'}
M'=M_{t}+(M-M_{t'})x^{t-t'}.
\end{equation}
We claim that $M'$ and $t$ satisfy the statement. Since $d\geq t'>t$, it follows from equation \eqref{equation:M'} that $d'=d-t'+t$ and therefore $s_1\leq t<t'=d+t-d'\leq s_2$.
Moreover
\begin{equation*}
(M'G)_t=(M_tG+(M-M_{t'})x^{t-t'}G)_t=(M_tG)_t=(MG)_t.	
\end{equation*}
Finally, 
\begin{equation*}
\begin{split}
&\Phi_{d'+\delta_1,d'+\delta_1-t-1}(M'G)=\Phi_{d'+\delta_1,d'+\delta_1-t-1}((M_{t}+(M-M_{t'})x^{t-t'})G)=\\
&=\Phi_{d'+\delta_1,d'+\delta_1-t-1}(M_{t}G)+\Phi_{d'+\delta_1,d'+\delta_1-t-1}(((M-M_{t'})x^{t-t'})G)=\\
&=\Phi_{t+\delta_1,\delta_1-1}(M_{t}G)+\Phi_{d'+\delta_1+t'-t,d'+\delta_1-t-1}((M-M_{t'})G)=\\
&=\Phi_{t+\delta_1,\delta_1-1}(M_{t}G)+\Phi_{d+\delta_1,d'+\delta_1-t-1}(MG)-\Phi_{t'+\delta_1,\delta_1-1}(M_{t'}G)=\\
&=\Phi_{d+\delta_1,d'+\delta_1-t-1}(MG),
\end{split}
\end{equation*}
where the second and third equalities follow from Lemma~\ref{lemma:translation}, while the last one follows from the fact that $\Phi_{t+\delta_1,\delta_1-1}(M_{t}G)=\Phi_{t'+\delta_1,\delta_1-1}(M_{t'}G)$.
\end{proof}

\begin{remark}
By Lemma~\ref{lemma:reduction} we have that $$|\supp(M'G)|\leq |\supp(MG)|.$$
Indeed
\begin{equation*}
\begin{split}
|\supp(M'G)|&=|\supp((M'G)_t)|+|\supp(\Phi_{d'+\delta_1,d'+\delta_1-t-1}(M'G))|\\
&=|\supp((MG)_t)|+|\supp(\Phi_{d+\delta_1,d'+\delta_1-t-1}(MG))|\\
&\leq |\supp(MG)|.
\end{split}
\end{equation*}

Since $G$ is a generator matrix of $\Cc$, the submodules spanned by the rows of $MG$ and $M'G$ are subcodes of $\Cc$. Suppose that the submodule spanned by the rows of $MG$ realizes $d_r(\Cc)$. Since Lemma~\ref{lemma:reduction} implies that $|\supp(M'G)|\leq |\supp(MG)|$, if we had $\rk(MG)=\rk(M'G)$, then $M'G$ would realize $d_r(\Cc)$, too. However, it may happen that $\rk(MG)\neq\rk(M'G)$, as the next example shows.
\end{remark}

\begin{example}
Let $\Cc$ be an $(3,2,1)$ code in $\F_2[x]^{3}$ with row-reduced generator matrix
\begin{equation*}
G=\begin{pmatrix}
1&0&x\\0&1&0
\end{pmatrix}.
\end{equation*}
Let $M\in\F_2[x]^{2\times 2}$ be the following matrix of degree $d=38$
\begin{equation*}
M=\begin{pmatrix}
x^{38}&x^{17}\\x^{17}&1
\end{pmatrix}. 
\end{equation*}
Let $s_1=17$ and $s_2=37$. Since $s_2-s_1=20\geq 16=2^4=q^{\delta_1kr}$, we can apply Lemma~\ref{lemma:reduction} to $M$ and $G$. Let $M'\in\F_2[x]^{2\times 2}$ be the matrix
\begin{equation*}
M'=\begin{pmatrix}
x^{34}&x^{17}\\x^{17}&1
\end{pmatrix}. 
\end{equation*}
We claim that $M'$ with $t=20$ satisfies the condition in Lemma~\ref{lemma:reduction}. Indeed,
we have that $s_1=17< t=20<24=38+20-34=d+t-d'<37$. Moreover, 
\begin{equation*}
(M'G)_{20}=\begin{pmatrix}
0&x^{17}&0\\x^{17}&1&x^{18}
\end{pmatrix}=(MG)_{20}. 
\end{equation*}
Finally, we have that
\begin{equation*}
\Phi_{35,14}(M'G)_{20}=\begin{pmatrix}
x^{13}&0&x^{14}\\0&0&0
\end{pmatrix}=\Phi_{39,14}(MG). 
\end{equation*}
Notice that $\det(M')=0$ while $\det(M)=x^{38}-x^{34}$, therefore $\rk(MG)\neq\rk(M'G)$. Notice moreover that $(M',20)$ is not the only pair that satisfies the conditions in Lemma~\ref{lemma:reduction} and that other pairs may behave differently. For example, consider the matrix $M''\in\F_2[x]^{2\times 2}$ given by
\begin{equation*}
M''=\begin{pmatrix}
x^{35}&x^{17}\\x^{17}&1
\end{pmatrix}. 
\end{equation*}
One can verify that the pair $(M'',20)$ satisfies the conditions in Lemma~\ref{lemma:reduction} and $\rk(MG)=\rk(M''G)$.
\end{example}

For a matrix $M\in\F_q^{m\times n}$ and $S\subseteq \{1,\dots,k\}$, $L\subseteq\{1,\dots.n\}$ we let $M(S,L)$ denote the submatrix of $M$ consisting of the rows indexed by $S$ and the columns indexed by $L$. We are now ready to prove the main result of this section. It says that, in order to compute generalized weights, it suffices to look at a subcodes generated by codewords of bounded degree. A similar result for the generalized Hamming weights of a convolutional code is stated without proof in~\cite{RY,Y}.

\begin{theorem}\label{theorem:bound}
Let $\Cc$ be an $(n,k,\delta)$ convolutional code with maximum degree $\delta_1$. Then, $d_r(\Cc)$ is realized by a subspace generated by codewords of degree at most $((r+2)q^{\delta_1kr}+1)(n(\delta_1+1)-k+r)+\delta_1$.
\end{theorem}

\begin{proof}
If $\delta_1=0$, $d_r(\Cc)$ is realized by a subspace generated by codewords of degree zero by Proposition~\ref{proposition:hammingweights}. If $\delta_1\geq 1$, let $G$ be a generator matrix for $\Cc$ with maximum degree $\delta_1$. Suppose that $d_r(\Cc)$ is realized by a subspace generated by the rows of $MG$, where $M=(m_{i,j}(x))_{i,j}\in\F_q[x]^{r\times k}$ is a matrix with $\max_{i,j}\deg(m_{i,j}(x))\geq ((r+2)q^{\delta_1kr}+1)(n(\delta_1+1)-k+r)+1$ and such that there exists $N\subseteq\{1,\dots n\}$ with $|N|=r$ and $MG[R,N]$ has non-zero determinant, where $R=\{1,\ldots,r\}$. 
By Proposition~\ref{propproperties} 
$$d_r(\Cc)\leq n(\delta_1+1)-k+r.$$
Therefore, there exist $s_1>s_2>\dots>s_{r+3}$ such that 
$$s_{h}-s_{h+1}\geq q^{\delta_1kr}\text{ and }\Phi_{s_1,s_1-s_{r+3}}(MG)=0.$$
By induction, we define a chain of matrices $M^{(1)},\dots,M^{(r+2)}$ as follows. Applying Lemma~\ref{lemma:translation} to $M,s_1,s_2$, we obtain $M^{(1)}=(m^{(1)}_{i,j}(x))_{i,j}\in\F_q[x]^{r\times k}$. In the same way, applying Lemma~\ref{lemma:translation} to $M^{(h)},s_{h},s_{h+1}$, we obtain a matrix $M^{(h+1)}=(m^{(h+1)}_{i,j}(x))_{i,j}\in\F_q[x]^{r\times k}$. By Lemma~\ref{lemma:translation}, we have
\begin{equation*}
(MG)_{s_{r+3}}=(M^{(1)}G)_{s_{r+3}}=\dots=(M^{(r+2)}G)_{s_{r+3}}
\end{equation*}
and
\begin{equation*}
|\supp(MG)|=|\supp(M^{(1)}G)|=\dots=|\supp(M^{(r+2)}G)|.
\end{equation*}
If there exists $1\leq \bar h\leq r+2$ such that $\det(M^{(\bar h)}G[R,N])\neq 0$, then $d_r(\Cc)$ is also realized by $MG^{(\bar h)}$ and we conclude since the maximum degree of $MG^{(\bar h)}$ is smaller than that of $MG$.
Else, compute the determinant of $MG[R,N]$, indicating the monomial operations without solving them.
We obtain
\begin{equation*}
\det(MG[R,N])=\sum_{u}a_ux^{\alpha_{1,u}}\dots x^{\alpha_{r,u}}.
\end{equation*}
We rewrite this sum as
\begin{equation*}
\det(MG[R,N])=\sum_{d,\ell}S(MG,d,\ell),
\end{equation*}
where with $S(MG,d,\ell)$ we denote the partial sum of those terms $a_ux^{\alpha_{1,u}}\dots x^{\alpha_{r,u}}$ that appears in $\det(MG[R,N])$ such that $\sum \alpha_{h,u}=d$ and the number of exponents that are greater or equal than $s_{r+3}$ is exactly $\ell$. Notice that there exists at least a pair $d,\ell$ such that $S(MG,d,\ell)\neq0$ since $\det(MG[R,N])\neq0$. In the same way we define $S(M^{(h)}G,d,\ell)$ for each $1\leq h\leq r+2$.  By construction, we have that for each $1\leq h\leq r+2$ there exists a natural number $w_h$ such that
\begin{equation}\label{equation:Sdl}
\det(M^{(h)}G[R,N])=\sum_{d,\ell}S(M^{(h)}G,d,\ell)=\sum_{d,\ell}S(MG,d,\ell)x^{-w_h(\ell)}.
\end{equation}
For each $1\leq h\leq r+2$, let $y_h$ be
\begin{equation*}
y_h=\max\{d:\text{ there exists }\ell \text{ such that }S(M^{(h)}G,d,\ell)\neq 0\}
\end{equation*}
and $z_h$ be
\begin{equation*}
z_h=\min\{\ell:S(M^{(h)}G,y_h,\ell)\neq 0\}.
\end{equation*}
We prove now that for $0\leq h\leq r+1$ we have that $z_h>z_{h+1}$. Since $\det(M^{(h+1)}G[R,N])=0$, there exists $\bar\ell>z_{h+1}$ such that $S(M^{(h+1)}G,y_{h+1},\bar \ell)\neq0$.  Then, by equation \eqref{equation:Sdl} we have that
\begin{equation*}
y_h -(w_{h+1}-w_h)\bar\ell\geq y_{h+1}\geq y_h -(w_{h+1}-w_h)z_h.
\end{equation*}
So, we obtain that $z_{h+1}<\bar\ell\leq z_h$. Therefore, we have that $r\geq z_1>z_2>\dots>z_{r+2}\geq 0$. Since this is a contradiction, we conclude that there exists $h\in\{1,\dots,r\}$ such that $\det(M^{(h)}G[R,N])\neq0$.
\end{proof}

Even though we do not expect the bound of Theorem~\ref{theorem:bound} to be sharp, the theorem implies that the generalized weights of a convolutional code can be computed by exhaustive search in a finite number of steps. However, since the upper bound in Theorem~\ref{theorem:bound} is large, it does not lead to a practical algorithm. It would be interesting to better understand under which assumptions and by how much this bound can be improved.

\begin{question}\label{question:boundweight}
Is it possible to sharpen the bound in Theorem~\ref{theorem:bound}?
\end{question}

In this paper we improve the bound of Theorem~\ref{theorem:bound} for two families of codes. Proposition~\ref{proposition:hammingweights} shows that, if $\delta=0$, then $d_r(\Cc)$ is realized by a subspace generated by codewords of degree $0$ for all $r$, while Proposition~\ref{proposition:boundsMDS} improves the bound of Theorem~\ref{theorem:bound} for certain MDS codes.

\bigskip
	
In order to further simplify the computation of the generalized weights, we show that they are realized by subspaces generated by elements of minimal support. This is analogous to what happens for generalized Hamming weights of linear block codes and in fact for a much larger class of codes and supports, as discussed in~\cite[Section 3]{GR22}.

\begin{lemma}\label{lemma:minsupp}
Let $\Cc$ be an $(n,k,\delta)$ convolutional code and let $u_1,\dots,u_r\in\Cc$ be such that $\rk\left(\langle u_1,\dots,u_r\rangle_{\F_q[x]}\right)=r$. Then there exist $u_1',\dots,u_r'\in\Cc$ minimal codewords such that $\supp(u_i')\subseteq\supp(u_i)$ for $1\leq i\leq r$ and $\rk\left(\langle u_1',\dots,u_r'\rangle_{\F_q[x]}\right)=r$.
\end{lemma}

\begin{proof}
Let $e\in\Cc$ be a minimal codeword with $\supp(e)\subseteq \supp(u_1)$. If $$\rk(\langle e,u_2,\dots,u_r\rangle_{\F_q[x]})=r,$$ then let $u_1'=e$. Else, we claim that there exists $\alpha\in\F_q^*$ such that $\supp(u_1-\alpha e)\subsetneq \supp(u_1)$ and $$\rk(\langle u_1-\alpha e,u_2,\dots,u_r\rangle_{\F_q[x]})=r.$$ In fact, if  $\rk(\langle u_1-\alpha e,u_2,\dots,u_r\rangle_{\F_q[x]})<r$, then there exist $p_1(x),\dots,p_r(x),q_1(x),\dots,q_r(x)\in\F_q[x]$ such that $p_1(x),q_1(x)\neq 0$ and
$$p_1(x)e=\sum_{i=2}^rp_i(x)u_i\text{ and } q_1(x)(u_1-\alpha e)=\sum_{i=2}^rq_i(x)u_i.$$
Therefore 
$$p_1(x)q_1(x)u_1=\sum_{i=2}^r (p_1(x)q_i(x)+\alpha q_1(x)p_i(x))u_i,$$
contradicting the assumption that $\rk\left(\langle u_1,\ldots,u_r\rangle_{\F_q[x]}\right)=r$. If $u_1-\alpha e$ is a minimal codeword, let $u_1'=u_1-\alpha e$, otherwise we repeat this process. Notice that, since at each step the support becomes strictly smaller, we find a minimal codeword in a finite number of steps. Proceeding in the same way for $u_2,\dots, u_r$ we find minimal codewords $u_1',\dots,u_r'\in\Cc$ with the desired properties.
\end{proof}

The next theorem shows that, when computing generalized weights, we may restrict to subcodes generated by minimal codewords with the property that the support of each of them is not contained in the union of the supports of the others.

\begin{theorem}\label{theorem:minsuppdisjoint}
Let $\Cc$ be an $(n,k,\delta)$ convolutional code, let $1\leq r\leq k$. Then there exist $r$ minimal codewords $u_1,\dots,u_r\in\Cc$ such that $\langle u_1,\ldots,u_r\rangle_{\F_q[x]}$ realizes $d_r(\Cc)$ and $\supp(u_i)\nsubseteq \bigcup_{j\neq i}\supp(u_j)$ for $1\leq i\leq r$.
\end{theorem}

\begin{proof}
Let $\mathcal{U}=\langle u_1,\dots,u_r\rangle_{\F_q[x]}\subseteq\Cc$ realize $d_r(\Cc)=|\supp(\{u_1,\dots,u_r\})|$. Up to performing Gaussian elimination, we can assume that $\supp(u_i)\nsubseteq \bigcup_{j\neq i}\supp(u_j)$ for $1\leq i\leq r$. By Lemma~\ref{lemma:minsupp} there exist $u_1',\dots,u_r'$ elements of minimal support such that $\supp(u_i')\subseteq \supp(u_i)$ for $1\leq i\leq r$ and $\rk\left(\mathcal{U}'\right)=r$, where $\mathcal{U}'=\langle u_1',\dots,u_r'\rangle_{\F_q[x]}$. By construction, 
\begin{equation}\label{eqn:wtU}
\wt(\mathcal{U})=|\supp(\{u_1,\dots,u_r\})|\geq|\supp(\{u_1',\dots,u_r'\})|\geq\wt(\mathcal{U}').
\end{equation} 
Since $d_r(\Cc)=\wt(\mathcal{U})$ and $\rk\left(\mathcal{U}'\right)=r$, we conclude that $\wt(\mathcal{U})=\wt(\mathcal{U}')$ and $\mathcal{U}'$ realizes $d_r(\Cc)$. In addition, $\supp(\{u_1,\ldots,u_r\})=\supp(\{u'_1,\ldots,u'_r\})$ by (\ref{eqn:wtU}).
Therefore
\begin{equation*}
\bigcup_{j\neq i}\supp(u_j)\subsetneq \bigcup_{j=1}^r\supp(u_j)=\bigcup_{j=1}^r\supp(u_j').
\end{equation*}
    Since by construction $\bigcup_{j\neq i}\supp(u_j')\subseteq\bigcup_{j\neq i}\supp(u_j)$, we conclude that $\supp(u_i')\nsubseteq\bigcup_{j\neq i}\supp(u_j')$.
\end{proof}

\section{MDS convolutional codes}

This section focuses on the generalized weights of MDS and MDP convolutional codes. Proposition~\ref{propproperties} yields lower bounds for the generalized weights of these two classes of codes. Indeed, for an MDS code $\Cc$ we know that \begin{equation}\label{eqn:MDSlowerbound}
d_i(\Cc)\geq  d_1(\Cc)+i-1=\df(\Cc)+i-1=(n-k)\left(\left\lfloor \frac{\delta}{k}\right\rfloor+1\right)+\delta+i.   
\end{equation}
Similarly, for an MDP code $\Cc$ we have that
\begin{equation*}
\begin{split}
d_i(\Cc)&\geq d_1(\Cc)+i-1=\df(\Cc)+i-1\geq d_{L}^c(\Cc)+i-1\\
&=(n-k)\left(\left\lfloor\frac{\delta}{k}\right\rfloor+\left\lfloor \frac{\delta}{n-k}\right\rfloor+1\right)+i
\end{split}
\end{equation*}
	
Proposition~\ref{propproperties} also provides an upper bound on the generalized weights of any convolutional code.	
In this section, we show that this bound can be sharpened when $\Cc$ is MDS or MDP. We start by showing that, if $\Cc$ is an MDS or MDP convolutional code, then the row degrees of a row-reduced matrix of $\Cc$ can only take certain values.
	
\begin{lemma}\label{propineq}
Let $\Cc$ be an $(n,k,\delta)$ a convolutional code, $G(x)=(p_{i,j}(x))_{i,j}$ be a row-reduced generator matrix for $\Cc$ and $\delta_i=\max_{j}\deg(p_{i,j}(x))$ for $1\leq i\leq k$. \begin{enumerate}
\item If $\Cc$ is MDP, then $\left\lfloor \frac{\delta}{k}\right\rfloor\leq\delta_i\leq \left\lfloor \frac{\delta}{k}\right\rfloor+ k-a$, where $\delta=k\left\lceil\frac{\delta}{k}\right\rceil-a$ and $0\leq a<k$.
\item If $\Cc$ is MDS, then
$\left\lfloor \frac{\delta}{k}\right\rfloor\leq \delta_i\leq \left\lfloor \frac{\delta}{k}\right\rfloor+1$.
\item If $\Cc$ is MDS or MDP and $k\mid \delta$, then $\delta_i=\frac{\delta}{k}$.
\end{enumerate}
\end{lemma}

\begin{proof}
Suppose by contradiction that there exists an index $i$ such that $\delta_i<\left\lfloor \frac{\delta}{k}\right\rfloor$.
If $k=n$ and $\Cc$ is MDS, then $$\wt(p_{i,1}(x),\dots,p_{i,n}(x))\leq n\left\lfloor \frac{\delta}{n}\right\rfloor<\delta+1=\df(\Cc).$$
If $k<n$, then
\begin{equation*}
\begin{split}
&\wt(p_{i,1}(x),\dots,p_{i,n}(x))\leq n\left\lfloor \frac{\delta}{k}\right\rfloor\\
&=n\left(\left\lfloor \frac{\delta}{k}\right\rfloor+1\right)-n=(n-k)\left(\left\lfloor \frac{\delta}{k}\right\rfloor+1\right)+k\left\lfloor \frac{\delta}{k}\right\rfloor+k-n\\
&<(n-k)\left(\left\lfloor \frac{\delta}{k}\right\rfloor+1\right)+(n-k)\left\lfloor \frac{\delta}{n-k}\right\rfloor+1.
\end{split}
\end{equation*}
If $\Cc$ is MDP, then 
$$(n-k)\left(\left\lfloor \frac{\delta}{k}\right\rfloor+1\right)+(n-k)\left\lfloor \frac{\delta}{n-k}\right\rfloor+1=d_L^c(\Cc)\leq \df(\Cc).$$ If $\Cc$ is MDS, then
$$(n-k)\left(\left\lfloor \frac{\delta}{k}\right\rfloor+1\right)+(n-k)\left\lfloor \frac{\delta}{n-k}\right\rfloor+1\leq\df(\Cc).$$
In all cases, the weight of the $i$-th row is strictly smaller than the free distance of the code, a contradiction. This proves that $\delta_i\geq\left\lfloor \frac{\delta}{k}\right\rfloor$ for $1\leq i\leq k$.

If $k\mid\delta$, then $\delta_i= \frac{\delta}{k}$, since $\sum \delta_i=\delta$ and $\delta_i\geq\frac{\delta}{k}$. 
Else, since $\delta_i\geq\left\lfloor \frac{\delta}{k}\right\rfloor$ and $\sum \delta_i=\delta$, then $\delta_i\leq \left\lfloor \frac{\delta}{k}\right\rfloor+k-a$. 

Finally, let $\Cc$ be MDS and suppose that there exists an index $\ell$ such that $\delta_{\ell}>\left\lfloor \frac{\delta}{k}\right\rfloor+1$. Consider the submatrix $G_{\ell}(x)$ obtained from $G(x)$ by deleting the $\ell$-th row and let $\Cc_{\ell}$ be the code associated to $G_{\ell}(x)$. Then $G_{\ell}(x)$ is row-reduced and $\Cc_{\ell}\subseteq \Cc$. By Proposition~\ref{propproperties} we have that $\df(\Cc_{\ell})\geq \df(\Cc)$. On the other hand, by Theorem~\ref{singletonbound} 
\begin{equation*}
\begin{split}
\df(\Cc_{\ell})&\leq (n-k+1)\left(\left\lfloor\frac{\delta-\left\lfloor\frac{\delta}{k}\right\rfloor-2}{k-1}\right\rfloor+1\right)+\delta-\left\lfloor\frac{\delta}{k}\right\rfloor-2+1\\
&=(n-k+1)\left(\left\lfloor\frac{\delta}{k}\right\rfloor+1\right)+\delta-\left\lfloor\frac{\delta}{k}\right\rfloor-1\\
&\leq (n-k)\left(\left\lfloor\frac{\delta}{k}\right\rfloor+1\right)+\delta=\df(\Cc)-1,
\end{split}
\end{equation*}
a contradiction.
\end{proof}

As a consequence of Lemma~\ref{propineq}, we obtain an upper bound for the generalized weights of an MDS code, that improves the bound 4. from Proposition~\ref{propproperties} whenever $k\nmid\delta$.

\begin{proposition}\label{propineq2}
Let $\Cc$ be an $(n,k,\delta)$ convolutional code and write $\delta=k\left\lceil\frac{\delta}{k}\right\rceil-a$ with $0\leq a<k$. 
\begin{enumerate}
\item If $\Cc$ is MDS, then
$$d_k(\Cc)\leq n\left(\left\lceil\frac{\delta}{k}\right\rceil+1 \right)-a.$$
\item If $\Cc$ is MDP, then 
$$d_k(\Cc)\leq n\left(\left\lfloor\frac{\delta}{k}\right\rfloor+k-a+1 \right)-k+1.$$
\end{enumerate}
\end{proposition}

\begin{proof}
1. If $a=0$, then the bound follows from Proposition~\ref{propproperties} and Lemma~\ref{propineq}. So, suppose that $a>0$. Let $G(x)$ be a row-reduced generator matrix for $\Cc$ and let $c_1(x),\dots,c_k(x)$ be the rows of $G(x)$. Since $\Cc$ is MDS, then $\deg(c_1(x))=\dots=\deg(c_{k-a}(x))=\left\lceil\frac{\delta}{k}\right\rceil$ and $\deg(c_{k-a+1}(x))=\dots=\deg(c_k)=\left\lfloor\frac{\delta}{k}\right\rfloor$. For $1\leq i\leq k-a$, there exist $\alpha_{i,1},\dots,\alpha_{i,a}\in\F_q$  such that the last $a$ entries of $\tilde c_i(x)=c_i(x)+\sum_{j=1}^a\alpha_{i,j}xc_{k-a+j}(x)$ have degree smaller or equal than $\left\lfloor\frac{\delta}{k}\right\rfloor$. Since 
$$\rk\left(\langle \tilde c_1(x),\dots,\tilde c_{k-a}(x),c_{k-a+1}(x),\dots,c_k(x)\rangle_{\F_q[x]}\right)=k,$$
we obtain that 
\begin{equation*}
d_k(\Cc)\leq |\supp(c_1(x),\dots,\tilde c_{k-a}(x),c_{k-a+1}(x),\dots,c_k(x))|\leq n\left(\left\lceil\frac{\delta}{k}\right\rceil+1\right)-a.
\end{equation*}

2. The proof is similar to that of part 1.
\end{proof}

The next corollary is a rewriting of the bound for MDS codes from Proposition~\ref{propineq2}.

\begin{corollary}\label{corollary:bound}
Let $\Cc$ be an $(n,k,\delta)$ MDS convolutional code.
\begin{enumerate}
\item If $k\mid\delta$, then
$d_k(\Cc)\leq (n-k)\left(\frac{\delta}{k}+1\right)+\delta +k$.
\item If $k\nmid \delta$, then $d_k(\Cc)\leq(n-k)\left(\left\lfloor \frac{\delta}{k}\right\rfloor+1\right)+\delta +n$.
\item If $k=n$, then $d_k(\Cc)\leq \delta+k$.
\end{enumerate}
\end{corollary}

A consequence of Corollary~\ref{corollary:bound} is that some, and in some cases all, of the generalized weights of an $(n,k,\delta)$ MDS code are determined by the code parameters.
In particular, the generalized weights of an MDS code such that $k\mid\delta$ meet bound 4. from Proposition~\ref{propproperties}.	
	
\begin{proposition}\label{prop:MDSbounds}
Let $\Cc$ be an $(n,k,\delta)$ MDS convolutional code.
\begin{enumerate}
\item If $k=n$, then $d_r(\Cc)=\delta+r$, for $1\leq r\leq k$.
\item If $k\mid\delta$, then
$$d_r(\Cc)=n\left( \frac{\delta}{k}+1\right)-k+r,$$
for $1\leq r\leq k$. 
\item If $\delta=k\left\lceil\frac{\delta}{k}\right\rceil-a$ with $0<a<k$, then 
$$d_r(\Cc)=(n-k)\left(\left\lfloor \frac{\delta}{k}\right\rfloor+1\right)+\delta +r,$$
for $1\leq r\leq a$.
\end{enumerate}
\end{proposition}

\begin{proof}
The equalities in 1. and 2. are equivalent to $d_k(\Cc)\leq\df(\Cc)+k-1$. The claims now follow from Proposition~\ref{propproperties}, more precisely from the fact that the generalized weights are increasing.

3. Let $G(x)$ be a row-reduced matrix for $\Cc$. Let $\Cc'\subseteq\Cc$ the subcode generated by the $a$ rows of $G(x)$ of degree $\left\lfloor\frac{\delta}{k}\right\rfloor$. Then $\Cc'$ is an MDS $\left(n,a,a\left\lfloor\frac{\delta}{k}\right\rfloor\right)$ convolutional code. Hence
$$(n-k)\left(\left\lfloor \frac{\delta}{k}\right\rfloor+1\right)+\delta +r\leq d_r(\Cc)\leq d_r(\Cc')=n\left(\left\lfloor\frac{\delta}{k}\right\rfloor+1\right)-a+r$$ for $1\leq r\leq a$,
where the first inequality follows from (\ref{eqn:MDSlowerbound}), the second from Proposition~\ref{propproperties}, and the equality from part 2. Since the first and last quantities agree, the thesis follows.
\end{proof}

Proposition~\ref{prop:MDSbounds} shows in particular that, if $k\nmid\delta$, then the first $a=k\left\lceil\frac{\delta}{k}\right\rceil-\delta$ generalized weights of an $(n,k,\delta)$ MDS convolutional code are determined by its parameters.
We now show that the other generalized weights are not in general determined by the parameters of the code. In the next example, we exhibit two MDS codes with parameters $(3,2,1)$ with different second generalized weight. In particular, by Proposition~\ref{proposition:isometric} this provides an example of MDS codes with the same parameters $(n,k,\delta)$, which are not isometric.

\begin{example}
(a) Let $\mathrm{char}(\F_q)\neq2,3$ and let $\Cc_1\subseteq\F_q[x]^3$ with be the code with generator matrix
$$\begin{pmatrix}
2x& x+1& x+1\\1&1&2
\end{pmatrix}.$$
The free distance of $\Cc_1$ is $d_{\mathrm{free}}(\Cc_1)=d_1(\Cc_1)= 3$, as shown in~\cite[Example 3.2]{JP20}. Since $$|\supp(\langle (2x,x+1, x+1),(x,x,2x)\rangle_{\F_q}|=5,$$ we have that $d_2(\Cc_1)\leq 5$. We claim that $d_2(\Cc_1)=5$. Assume by contradiction that there are $c_1,c_2\in\Cc_1$ such that $|\supp(\langle c_1,c_2\rangle_{\F_q})|\leq 4$ and $\rk(\langle c_1,c_2 \rangle_{\F_q[x]})=2$. Then there exists $c\in \langle c_1,c_2\rangle_{\F_q}$ such that $\wt(c)=3$ and $c$ has a zero entry. Let $p(x),q(x)\in\F_q[x]$ such that 
$$c=(p(x)+2xq(x),p(x)+q(x)(x+1),2p(x)+q(x)(x+1)).$$
Since $c$ has a zero entry, then $p(x),q(x)\neq0$. Moreover:
\begin{itemize}
\item If $p(x)+2xq(x)=0$, then 
$$\wt(c)=\wt(q(x)(-x+1))+\wt(q(x)(1-3x))=4.$$
\item If $p(x)+(x+1)q(x)=0$, then 
$$\wt(c)=\wt(q(x)(x-1))+\wt(-q(x)(x+1))=4.$$
\item If $2p(x)+q(x)(x+1)=0$, then
$$\wt(c)=\wt\left(\frac{1}{2}q(x)(3x-1)\right)+\wt\left(\frac{1}{2}q(x)(x+1) \right)=4.$$
\end{itemize}
Therefore $\wt(c)=4$, which yields a contradiction. We conclude that $d_2(\Cc_1)=5$.

\bigskip
	 	 
(b) Let $\mathrm{char}(\F_q)\neq2,3$ and let $\Cc_2\in\F_q[x]^3$ be the code with generator matrix
$$\begin{pmatrix}
2x& x+1& 0\\1&1&2
\end{pmatrix}.$$
It is easy to check that $d_1(\Cc_2)=3$. 
Moreover 
$$|\supp(\langle(2x,x+1,0),(x,x,2x)\rangle_{\F_q})|=4,$$ hence $d_2(\Cc_2)=4$ by Proposition \ref{propproperties}.
\end{example}

Next we produce a new upper bound for the last $k-a$ generalized weights of an MDS convolutional code, in the case when they are not determined by the parameters of the code.
To an $(n,k,\delta)$ convolutional code $\Cc=\langle c_1,\dots,c_k\rangle_{\F_q[x]}$ we associate the linear block code $\Cc[0]=\langle c_1[0],\dots,c_k[0]\rangle_{\F_q}$. Notice that, if $\Cc$ is MDS, then $\dim(\Cc[0])=k$. In the next proposition, we establish a relation between the last $k-a$ generalized weights of $\Cc$ and the first $k-a$ generalized Hamming weights of $\Cc[0]$. 

\begin{proposition}\label{proposition:MDSweightsbound}
Let $\Cc$ be an $(n,k,\delta)$ MDS convolutional code. If  $\delta=k\left\lceil\frac{\delta}{k}\right\rceil-a$ with $0<a<k$, then
\begin{equation*}
d_{a+r}(\Cc)\leq (n-k)\left(\left\lfloor\frac{\delta}{k}\right\rfloor+1\right)+\delta +a+\min\left\{d^{H}_r(\Cc[0]),d^{H}_r(\rev(\Cc)[0])\right\}
\end{equation*}
for $1\leq r\leq k-a$.
\end{proposition}

\begin{proof}
We start by showing that 
$$d_{a+r}(\Cc)\leq (n-k)\left(\left\lfloor\frac{\delta}{k}\right\rfloor+1\right)+\delta +a+d^{H}_r(\Cc[0]).$$
Since $\Cc$ is MDS, by Lemma \ref{propineq} there exist $c_1,\dots, c_k$ such that $\Cc=\langle c_1,\dots, c_k\rangle_{\F_q[x]}$, $\deg(c_1)=\dots=\deg(c_{a})=\left\lfloor\frac{\delta}{k}\right\rfloor$ and any element of $\langle c_{a+1},\dots,c_k\rangle_{\F_q[x]}$ has degree bigger than or equal to $\left\lfloor\frac{\delta}{k}\right\rfloor+1$.
If $c\in\langle c_1,\dots, c_a\rangle_{\F_q}$, then
$$wt(c[0])\geq \df(\Cc)-n\left\lfloor\frac{\delta}{k}\right\rfloor=n-a+1.$$
Since $d_{k-a}^H(\Cc[0])\leq n-a$, if $U\subseteq\Cc[0]$ realizes $d_r^H(\Cc[0])$, then $U\cap\langle c_1[0],\dots, c_a[0]\rangle_{\F_q}=0$. Therefore, there exist $c'_1,\dots,c'_r\in\Cc$ with $\deg(c'_1),\dots,\deg(c'_r)=\left\lceil\frac{\delta}{k}\right\rceil$ such that $U=\langle c'_1[0],\dots, c'_r[0]\rangle_{\F_q}$.
Let $\mathcal{D}=\langle c'_1,\dots,c'_r,xc_1,\dots,xc_a\rangle_{\F_q[x]}$. Then  $$\rk\left(\mathcal{D}\right)=\rk\left(\langle c'_1,\dots,c'_r,c_1,\dots,c_a\rangle_{\F_q[x]}\right)\geq\dim\left(\langle c'_1[0],\dots,c'_r[0],c_1[0],\dots,c_a[0]\rangle_{\F_q}\right)=a+r.$$ Therefore $\rk(\mathcal{D})=a+r$ and 
\begin{equation*}
d_{a+r}(\Cc)\leq\wt(\mathcal{D})\leq|\supp(\{ c'_1,\dots,c'_r,xc_1,\dots,xc_a\})|\leq(n-k)\left\lceil\frac{\delta}{k}\right\rceil+\delta +a+d^{r}_H(\Cc[0]).
\end{equation*}
In fact, $c'_1[0],\dots,c'_r[0]$ contribute $d^{r}_H(\Cc[0])$ and the rest at most $n\left\lceil\frac{\delta}{k}\right\rceil=(n-k)\left\lceil\frac{\delta}{k}\right\rceil+\delta +a$.
Since the same inequality holds for $\rev(\Cc)$, we conclude by Proposition~\ref{proposition:reverse}.
\end{proof}
	
\begin{question}
Is the bound in Proposition~\ref{proposition:MDSweightsbound} sharp, for any choice of the code parameters?
\end{question}

A positive answer to this question would complete the classification of the generalized weights of MDS codes. Moreover, it would imply that the generalized weights of an MDS code are realized by subspaces generated by codewords of degree at most $\left\lceil\frac{\delta}{k}\right\rceil$, thus answering Question~\ref{question:boundweight} for this specific class of codes. We conclude this section with a result in this direction.

\begin{proposition}\label{proposition:boundsMDS}
Let $\Cc$ be an $(n,k,\delta)$ sMDS convolutional code such that $\rev(\Cc)$ is MDP. Then, $d_r(\Cc)$ is realized by a subcode generated by codewords of degree smaller than $\left\lfloor\frac{\delta}{k}\right\rfloor+\left\lceil\frac{\delta}{n-k}\right\rceil+1$. 
\end{proposition}
	
\begin{proof}
By Theorem~\ref{theorem:minsuppdisjoint} we have that $d_r(\Cc)$ is realized by a subcode $\mathcal{D}$ generated by minimal codewords $c_1,\dots,c_r$ such that $\supp(c_i)\subsetneq\bigcup_{j\neq i} \supp(c_j)$ for $1\leq i\leq r$. If there exists $i$ such that $c_i$ has degree greater or equal to $\left\lfloor\frac{\delta}{k}\right\rfloor+\left\lceil\frac{\delta}{n-k}\right\rceil+1$, then 
$$\wt(c_i)\geq \df(\Cc)+n-k+1,$$
since $\Cc$ is sMDS and its reverse is MDP. We conclude that
$$\wt(\mathcal{D})\geq \df(\Cc)+n-k+r=(n-k)\left(\left\lfloor\frac{\delta}{k}\right\rfloor+1\right)+\delta+n+1+r-k.$$
This contradicts bound 1. or 2. in Corollary~\ref{corollary:bound}, since $\mathcal{D}$ realizes $d_r(\Cc)$.
\end{proof}

\section{Optimal anticodes}

In this section we prove an anticode bound for convolutional codes and define optimal convolutional anticodes as the codes that meet the anticode bound. We give a complete classification of optimal anticodes and we compute their generalized weights. We start by recalling the anticode bound for linear block codes with the Hamming metric.

\begin{theorem}\label{theorem:anticodeboundHamming}
Let $\Cc\subseteq \F_q^n$ be an $\F_q$-linear code. Then
$$\dim(\Cc)\leq\mathrm{maxwt}_H(\Cc).$$
\end{theorem}

A convolutional code contains codewords of arbitrarily large weight. However, any finite dimensional subspace contains a finite number of codewords, therefore one can define its maximum weight in the usual way.

\begin{definition}
Let $\Cc\subseteq\F_q[x]^n$ be an $(n,k,\delta)$ convolutional code and let $U\subseteq\Cc$ be a finite dimensional vector space. The {\bf maximum weight} of $U$ is
$$\maxwt(U)=\max\{\wt(u)\mid u\in U\}.$$
The {\bf maximum weight} of $\Cc$ is $$\maxwt(\Cc)=\min\left\{\maxwt(U):U \text{ is an }\F_q\text{-linear subspace of }\Cc\text{ and }\rk\left(\langle U\rangle_{\F_q[x]}\right)=\rk(\Cc)\right\}.$$
\end{definition}

\begin{remark}
For an $(n,k,\delta)$ convolutional code $\Cc$ one has
$$\maxwt(\Cc)=\min\left\{\maxwt(U):U \text{ is an }\F_q\text{-linear subspace of }\Cc\text{ and }\dim(U)=k=\rk\left(\langle U\rangle_{\F_q[x]}\right)\right\}.$$
In fact, every $\F_q$-linear subspace $V\subseteq\Cc$ which generates a subcode of rank $k$ contains a $k$-dimensional $\F_q$-linear subspace $U\subseteq V$ such that $U$ generates a subcode of rank $k$ and $\wt(U)\leq\wt(V)$.
\end{remark}

As a consequence of Theorem~\ref{theorem:anticodeboundHamming} we obtain a bound for convolutional codes that we call anticode bound, in analogy with the Hamming metric case and the rank-metric case.

\begin{theorem}[Anticode bound]\label{theorem:anticodebound}
Let $\Cc\subseteq\F_q[x]^n$ be an $(n,k,\delta)$ convolutional code. Then
\begin{equation*}
\rk(\Cc)\leq\maxwt(\Cc).
\end{equation*}
\end{theorem}

\begin{proof}
Let $U$ be an $\F_q$-linear subspace of $\Cc$ such that $\rk(\langle U\rangle_{\F_q[x]})=k$ and let $d=\max\{\deg(c):c\in U\}$. There exists an $\F_q$-linear injective homomorphism $\varphi:U\rightarrow \F_q^{n(d+1)}$ such that $\wt(c)=\wt_H(\phi(c))$ for all $c\in\Cc$. Since $\dim(U)=\dim(\varphi(U))$ and $\maxwt(U)=\maxwt_H(\varphi(U))$, by Theorem~\ref{theorem:anticodeboundHamming} we conclude that
\begin{equation*}
\rk(\Cc)=\rk\left(\langle U\rangle_{\F_q[x]}\right)\leq\dim(U)\leq\mathrm{maxwt}_H(\varphi(U))=\maxwt(U).
\end{equation*}
Since the inequality holds for every $\F_q$-linear subspace $U\subseteq\Cc$, we conclude.
\end{proof}
	
\begin{definition}
A convolutional code $\Aa$ is an ${\bf optimal\ (convolutional)\ anticode}$ if $$\rk(\Aa)=\maxwt(\Aa).$$
\end{definition}

The next proposition shows that two optimal anticodes with the same rank have the same generalized weights, when $q\neq2$.

\begin{proposition}\label{proposition:weightsoptimal1}
Let $q\neq2$ and let $\Aa\subseteq \F_q[x]^n$ be an optimal anticode. Then $d_r(\Aa)=r$ for $1\leq r\leq\rk(\Aa)$.
\end{proposition}

\begin{proof}
Since $\Aa$ is an optimal anticode, there exists an $\F_q$-linear subspace $U$ such that $\rk(\Aa)=\rk\left( \langle U\rangle_{\F_q[x]}\right)=\maxwt(U)$.
Let $\varphi$ be as in the proof of Theorem~\ref{theorem:anticodebound}. By Theorem~\ref{theorem:anticodeboundHamming}  
$$\rk\left( \langle U\rangle_{\F_q[x]}\right)\leq\dim(\varphi(U))\leq\maxwt_H(\varphi(U))=\maxwt(U).$$
Then $\varphi(U)$ is an optimal anticode in the Hamming metric and $|\supp(\varphi(U))|=\maxwt_H(\varphi(U))$, since $q\neq 2$. Since $|\supp(U)|=|\supp(\varphi(U))|$, we have that
\begin{equation*}
|\supp(U)|=\maxwt(U)=\dim(U)=\rk(\Aa).
\end{equation*}
This implies that $d_{\rk(\Aa)}(\Aa)\leq\rk(\Aa)$. Since $1\leq d_1(\Aa)<\ldots<d_{\rk(\Aa)}(\Aa)\leq\rk(\Aa)$ by Proposition~\ref{propproperties}, we conclude.
\end{proof}

The assumption that $q\neq2$ in the previous proposition is necessary, as the next example shows
	
\begin{example}\label{ex:OACq=2}
Let $\Aa=\langle (1,1,0),(1,0,1)\rangle_{\F_2[x]}$. $\mathcal{A}$ is an optimal anticode, since $\rk(\mathcal{A})=2=\maxwt_H(\langle (1,1,0),(1,0,1)\rangle_{\F_2})\geq\maxwt(\mathcal{A})$. However, $d_1(\Aa)=2$ and $d_2(\Aa)=3$ by Proposition \ref{proposition:hammingweights}.  
\end{example}

The next result is the converse of Proposition~\ref{proposition:weightsoptimal1}. Notice that, in this case, we do not need to assume that $q\neq2$.

\begin{proposition}\label{proposition:weightsoptimal2}
Let $\Aa\subseteq\F_q[x]^n$ be an $(n,k,\delta)$ convolutional code. If $d_{k}(\Aa)=k$, then $\Aa$ is an optimal anticode.
\end{proposition}

\begin{proof}
Let $U$ be an $\F_q$-linear subspace of $\Aa$ such that $d_{k}(\Aa)=|\supp(U)|$ and $\rk\left(\langle U\rangle_{\F_q[x]}\right)=k$. By Theorem~\ref{theorem:anticodeboundHamming}
\begin{equation*}
\maxwt(U)\leq|\supp(U)|=k\leq \dim(U)\leq\maxwt(U).
\end{equation*}
Therefore $\maxwt(U)=k\geq\maxwt(\Aa)$, so $\Aa$ is an optimal anticode by Theorem~\ref{theorem:anticodebound}.
\end{proof}

Because of the generalized-weight-preserving correspondence between $(n,k,0)$ convolutional codes and $(n,k)$ linear block codes, the next result is not surprising.

\begin{proposition}
An $(n,k,0)$ convolutional code $\Cc\subseteq\F_q[x]^n$ is an optimal anticode if and only if $\Cc[0]\subseteq\F_q^n$ is an optimal anticode with respect to the Hamming metric.
\end{proposition}

\begin{proof}
If $\Cc[0]\subseteq\F_q^n$ is an optimal anticode with respect to the Hamming metric, then 
$$\maxwt(\Cc)\leq\maxwt(\Cc[0])=\dim(\Cc[0])=k,$$
where the last equality follows from Proposition~\ref{prop:C[0]}. Therefore, $\Cc$ is an optimal anticode.

Conversely, suppose that $\Cc$ is an optimal anticode. If $q\neq2$, by Proposition~\ref{proposition:hammingweights} and Proposition~\ref{proposition:weightsoptimal1} we have that $$\dim(\Cc[0])=k=d_k(\Cc)=d_k^H(\Cc[0])=\lvert\supp(\Cc[0])\rvert\geq\maxwt(\Cc[0]).$$ Therefore $\dim(\Cc[0])=\maxwt(\Cc[0])$ by Theorem~\ref{theorem:anticodebound} and $\Cc[0]$ is an optimal anticode. Let $q=2$. Since $\Cc$ is an optimal anticode and by Theorem~\ref{theorem:minsuppdisjoint}, there exists $U=\langle c_1(x),\dots,c_k(x)\rangle_{\F_q}$ such that $k=\rk\left(\langle U\rangle_{\F_q[x]}\right)=\dim(U)=\maxwt(U)$ and $\supp(c_j)\nsubseteq\bigcup_{i\neq j}\supp(c_i(x))$ for $1\leq j\leq k$. This implies that $|\supp(\sum_{i\neq j}c_i(x))|\geq k-1$ for $1\leq j\leq k$. Moreover,  $|\supp(\sum c_i(x))|\geq k$, therefore equality must hold. It follows that every element of $\supp(c_j)$, except for the one that belongs to no $\supp(c_i)$ for $i\neq j$, must belong to $\supp(\sum_{i\neq j} c_i(x))$. This implies that $\wt(c_j(x))\leq 2$ for $1\leq j\leq k$. Let
$$V=\{c_i[0],\dots,c_i[\deg(c_i)]: 1\leq i\leq k\}.$$
Since $c_i(x)\in\Cc$ for all $i$, then $c_i[j]\in\Cc[0]$ for $0\leq j\leq \deg(c_i)$. Therefore $V\subseteq\Cc[0]$. Moreover, $$k=\rk(\langle U\rangle_{\F_q[x]})\leq\rk(\langle V\rangle_{\F_q[x]})\leq\rk(\Cc)=k$$
from which $$k=\rk(\langle V\rangle_{\F_q[x]})\leq\dim(\langle V\rangle_{\F_q})\leq\dim(\Cc[0])=k.$$
It follows that $V=\Cc[0]$. Therefore, there exist $a_1,\dots,a_k\in \langle V\rangle_{\F_q}$ that satisfy the following conditions:
\begin{itemize}
    \item $\langle a_1,\dots,a_k\rangle_{\F_q}=\Cc[0]$,
    \item $\wt(a_1)=\dots=\wt(a_r)=1$ and $d_1^H(\langle a_{r+1},\dots,a_k\rangle_{\F_q})=2$ for some $1\leq r\leq k$,
    \item $\bigcup_{i=1}^r\supp(a_i)\cap\bigcup_{i=r+1}^k\supp(a_i)=\emptyset$,
    \item for $r+1\leq j\leq k$, there exist $1\leq \bar j\leq k$ and $0\leq \hat j\leq \deg(c_{\bar j})$ such that $a_j=c_{\bar j}\left[\hat j\right]$.
\end{itemize} 
Since $\wt(a_j)=2$ for $r+1\leq j\leq k$, up to permuting $c_1(x),\ldots,c_k(x)$, we may assume without loss of generality that $c_j(x)=c_j[\deg(c_j)]x^{\deg(c_j)}$ for $r+1\leq j\leq k$. In particular, $a_j=c_j[\deg(c_j)]$. Fix $1\leq i\leq r$. Suppose that  $\supp(\langle \{c_i[0],\dots,c_i[\deg(c_i)]\}\rangle)\subseteq\supp(\langle a_{r+1},\ldots,a_k\rangle)$, then $\langle c_i(x),c_{r+1}(x),\dots, c_k(x)\rangle_{\F_q[x]}\subseteq \langle a_{r+1},\dots,a_{k}\rangle_{\F_q[x]}$, but this is a contradiction since $\rk(\langle c_i(x),$ $c_{r+1}(x),\dots, c_k(x)\rangle_{\F_q[x]})=k-r+1$ while $\rk(\langle a_{r+1},\dots,a_{k}\rangle_{\F_q[x]})=k-r$. If instead we assume that $\lvert\supp(\langle \{c_i[0],\dots,c_i[\deg(c_i)]\}\rangle)\cap\supp(\langle a_{r+1},\ldots,a_k\rangle)\rvert=1$, we again find a contradiction since $d_1^H(\langle a_{r+1},\dots,a_k\rangle_{\F_q})=2$. Therefore $\supp(\langle \{c_i[0],\dots,c_i[\deg(c_i)]: 1\leq i\leq r\}\rangle)\subseteq\supp(\langle a_1,\ldots,a_r\rangle)$, hence $\bigcup_{i=1}^r\supp(c_i(x))\cap\bigcup_{i=r+1}^k\supp(c_i(x))=\emptyset$.

Suppose by contradiction that $\maxwt(\Cc[0])\geq k+1$. Then there exists $I\subseteq\{1,\dots,k\}$ such that
$$k+1\leq \wt\left(\sum_{i\in I}a_i\right)=\wt\left(\sum_{i\in S}a_i\right)+\wt\left(\sum_{i\in I\setminus S}a_i\right),$$
where $S=I\cap\{r+1,\dots,k\}$. Since $\wt(\sum_{i\in I\setminus S}a_i)=\lvert I\setminus S\rvert\leq r$, then $\wt(\sum_{i\in S}a_i)\geq k-r+1$. Therefore
\begin{equation*}
    \wt\left(\sum_{i=1}^rc_i(x)+\sum_{i\in S}c_i(x)\right)=\wt\left(\sum_{i=1}^rc_i(x)\right)+\wt\left(\sum_{i\in S}c_i(x)\right)\geq r+\wt\left(\sum_{i\in S}a_i\right)\geq k+1,
\end{equation*}
contradicting the assumption that $\Cc$ is an optimal anticode of rank $k$. Therefore we conclude that $\maxwt(\Cc[0])=k$, that is, $\Cc[0]$ is an optimal anticode.
\end{proof}	
	
The rest of the section is devoted to classifying optimal anticodes. We start by introducing the concept of elementary optimal anticode.

\begin{definition}
A code $\Aa\subseteq\F_q[x]^n$ with $\rk(\Aa)=k$ is an \textbf{elementary optimal anticode} if there exist a set $J=\{j_1,\dots,j_k\}$ with $1\leq j_1<\dots<j_k\leq n$ and non-negative integers $a_1,\dots,a_k$ such that $\Aa=\langle c_1(x),\dots,c_k(x)\rangle_{\F_q[x]}$, where for $1\leq i\leq k$ the only nonzero entry of $c_i(x)$ is $x^{a_i}$ in position $j_i$.
\end{definition}
	
\begin{lemma}\label{lemma:elementary}
Every elementary optimal anticode is an optimal anticode.
\end{lemma}
	
\begin{proof}
Let $\Aa$ be an elementary optimal anticode. By definition $\Aa=\langle c_1(x),\dots,c_k(x)\rangle_{\F_q[x]}$, where for $1\leq i\leq k$ the only nonzero entry of $c_i(x)$ is $x^{a_i}$ in position $j_i$. Let $U=\langle c_1(x),\dots,c_k(x)\rangle_{\F_q}$. Then, $\maxwt(\Aa)\leq\maxwt(U)=k=\rk(\Aa)$. We conclude by Theorem~\ref{theorem:anticodebound}.
\end{proof}

\begin{theorem}[Characterization of optimal anticodes]\label{theorem:classification}
Let $q\neq 2$ and let $\Aa\subseteq\F_q[x]^n$ be a code. Then:
$\Aa$ is an optimal anticode if and only if there exists an elementary optimal anticode $\Aa'$ such that $\rk(\Aa')=\rk(\Aa)$ and $\Aa'\subseteq \Aa$.
\end{theorem}
	
\begin{proof}
$\Leftarrow)$ If $\Aa$ contains an elementary optimal anticode $\Aa'$ of the same rank, then by Proposition~\ref{propproperties}, Lemma~\ref{lemma:elementary}, and Proposition \ref{proposition:weightsoptimal1}
$$d_{\rk(\Aa)}(\Aa)\leq d_{\rk(\Aa')}(\Aa')=\rk(\Aa')=\rk(\Aa).$$
Therefore $\Aa$ is an optimal anticode by Proposition~\ref{proposition:weightsoptimal2}.

$\Rightarrow)$ Assume that $\Aa$ is an optimal anticode. By Proposition~\ref{proposition:weightsoptimal1} we have that $d_{\rk(\Aa)}(\Aa)=\rk(\Aa)$. Then, there exists $U$ such that $|\supp(U)|=\dim(U)=\rk\left( \langle U\rangle_{\F_q[x]}\right)=\rk(\Aa)$. Therefore, $U$ is generated by elements of rank 1 which are supported on different entries. The generators of $U$ as an $\F_q$-linear space generate the elementary optimal anticode $\Aa'$ as an $\F_q[x]$-module.
\end{proof} 

Notice that not all optimal anticodes are elementary, as the next example shows.

\begin{example}
Let $\Aa=\langle (1,x),(x,0)\rangle_{\F_q[x]}$ be a code. It is easy to show that $\Aa$ is not an elementary optimal anticode. However $\Aa'=\langle(x,0),(0,x^2)\rangle_{\F_q[x]}\subseteq \Aa$ is an elementary optimal anticode with generalized weights $d_1(\Aa')=1$ and $d_2(\Aa')=2$. By Proposition \ref{propproperties}, $\Aa$ has the same generalized weights as $\Aa'$, so $\Aa$ is an optimal anticode by Proposition \ref{proposition:weightsoptimal2}.	
\end{example}

We conclude this section with a proof that the dual of an optimal anticode is an optimal anticode, provided that $q\neq 2$.

\begin{lemma}\label{lemma:supportoptimal}
Let $q\neq 2$. Every optimal anticode $\Aa\subseteq\F_q[x]^n$ with $\rk(\Aa)=r$ is contained in a code generated over $\F_q[x]$ by $r$ vectors of the standard basis of $\F_q^n$.
\end{lemma}
	
\begin{proof}
Since $\Aa$ is an optimal anticode, by Theorem~\ref{theorem:classification} there exists an elementary optimal anticode $\Aa'\subseteq \Aa$ such that $\rk(\Aa')=\rk(\Aa)=r$. We may assume without loss of generality that $\Aa'$ is maximal with respect to inclusion among the codes with those properties. By definition $\Aa'$ has a system of generators consisting of $r$ vectors of weight 1, let $i_1,\ldots,i_r$ be the positions of the nonzero entries in the generators of $\Aa'$. Suppose that there exists $c\in\Aa\setminus\mathcal{Aa}'$ which has a nonzero entry in a position different from $i_1,\ldots,i_r$ and let $\mathcal{B}=\Aa'+\langle c\rangle_{\F_q[x]}$. Then $\mathcal{B}\subseteq\Aa$ and $\rk(\mathcal{B})=\rk(\Aa)+1$. This is a contradiction, showing that $\Aa\subseteq\langle e_{i_1},\dots,e_{i_r}\rangle_{\F_q}[x]$, where $e_i$ denotes the $i$-th standard basis vector of $\F_q^n$.
\end{proof}

\begin{corollary}\label{corollary:dualoptimal}
Let $q\neq 2$. The dual code $\Aa^{\perp}$ of an optimal anticode $\Aa$ is an elementary optimal anticode generated by vectors of the standard basis of $\F_q^n$. 
\end{corollary}

\begin{proof}
Let $\rk(\Aa)=r$, then $\rk(\Aa^{\perp})=n-r$. 
By Lemma~\ref{lemma:supportoptimal}, $\Aa$ is generated by $c_1,\ldots,c_r$ with $\wt(c_j)=1$ for $1\leq j\leq r$. Let $i_j\in\{1,\ldots,n\}$ be the position of the nonzero entry of $c_j$, $1\leq j\leq r$ and let $J=\{1,\dots,n\}\setminus \{i_1,\dots,i_r\}$. Consider the elementary optimal anticode $\mathcal{B}=\langle \{e_j\}_{j\in J}\rangle_{\F_q[x]}$. It is easy to check that $\mathcal{B}\subseteq\Aa^{\perp}$. Moreover, any module that properly contains $\mathcal{B}$ has rank larger than $n-r$. We conclude that $\Aa^{\perp}=\mathcal{B}$.
\end{proof}

\begin{corollary}\label{cor:noncatOAC}
Let $q\neq 2$. An optimal anticode $\Aa\subseteq\F_q[x]^n$ of rank $r$ is noncatastrophic if and only if is generated  by $r$ vectors of the standard basis of $\F_q^n$.
\end{corollary}

\begin{proof}
A code generated by $r$ vectors of the standard basis of $\F_q^n$ is an elementary optimal anticode and it is noncatastrophic by Proposition~\ref{proposition:noncatastrophic}.
Conversely, let $\Aa$ be a noncatastrophic optimal anticode. Then $\Aa=(\Aa^{\perp})^{\perp}$ and we conclude by Corollary~\ref{corollary:dualoptimal}.
\end{proof}

The conclusions of Lemma~\ref{lemma:supportoptimal}, Corollary~\ref{corollary:dualoptimal}, and Corollary~\ref{cor:noncatOAC} do not hold over $\F_2$, as the next example shows.

\begin{example}
Let $\Aa=\langle(1,1,0),(1,0,1)\rangle_{\F_2[x]}$ be the optimal anticode from Example~\ref{ex:OACq=2}. It is clear that $\Aa$ is not contained in any subcode of $\F_2[x]^3$ generated by two vectors of the standard basis of $\F_2^3$. Moroever, it is easy to show that $\Aa^{\perp}=\langle(1,1,1)\rangle_{\F_2[x]}$, in particular $\Aa^{\perp}$ is not an elementary optimal anticode. Finally, $\Aa$ is noncatastrophic by Proposition~\ref{proposition:noncatastrophic}, but it does not contain any vector of weight 1.
\end{example}

\bibliographystyle{plain}	
\bibliography{convgenwbib}

\begin{thebibliography}{10}

\bibitem{CGLLMS}
Eduardo Camps-Moreno, Elisa Gorla, Cristina Landolina, Elisa~Lorenzo García,
  Umberto Martínez-Peñas, and Flavio Salizzoni.
\newblock Optimal anticodes, {MSRD} codes, and generalized weights in the
  sum-rank metric.
\newblock {\em IEEE Transactions on Information Theory}, 68(6):3806--3822,
  2022.

\bibitem{CFN17}
S.~D. Cardell, M.~Firer, and D.~Napp.
\newblock Generalized column distances for convolutional codes.
\newblock In {\em 2017 IEEE International Symposium on Information Theory
  (ISIT)}, pages 21--25, 2017.

\bibitem{CFN20}
S.~D. Cardell, M.~Firer, and D.~Napp.
\newblock Generalized column distances.
\newblock {\em IEEE Transactions on Information Theory}, 66(11):6863--6871,
  2020.

\bibitem{CFN19}
S.~D. Cardell, D.~Napp, and M.~Firer.
\newblock Unrestricted generalized column distances: A wider definition.
\newblock In {\em 2019 IEEE International Symposium on Information Theory
  (ISIT)}, pages 2783--2787, 2019.

\bibitem{For}
G.~D. Forney.
\newblock Dimension/length profiles and trellis complexity of linear block
  codes.
\newblock {\em IEEE Transactions on Information Theory}, 40(6):1741--1752,
  1994.

\bibitem{Glu}
H.~Gluesing-Luerssen.
\newblock On isometries for convolutional codes.
\newblock {\em Advances in Mathematics of Communications}, 3:179--203, 2009.

\bibitem{LRS}
H.~Gluesing-Luerssen, J.~Rosenthal, and R.~Smarandache.
\newblock Strongly-{MDS} convolutional codes.
\newblock {\em IEEE Transactions on Information Theory}, 52(2):584--598, 2006.

\bibitem{GR22}
E.~Gorla and A.~Ravagnani.
\newblock Generalized weights of codes over rings and invariants of monomial
  ideals.
\newblock Preprint available at \url{https://arxiv.org/abs/2201.05813}, 2022.

\bibitem{HKM}
T.~Helleseth, T.~Kl{\o}ve, and J.~Mykkeltveit.
\newblock The weight distribution of irreducible cyclic codes with block
  lengths $n_1((q^{\ell}-1)/n)$.
\newblock {\em Discrete Math.}, 18:179--211, 1977.

\bibitem{HS}
H.~Horimoto and K.~Shiromoto.
\newblock On generalized {H}amming weights for codes over finite chain rings.
\newblock In Serdar Bozta{\c{s}} and Igor~E. Shparlinski, editors, {\em Applied
  Algebra, Algebraic Algorithms and Error-Correcting Codes}, pages 141--150,
  Berlin, Heidelberg, 2001. Springer Berlin Heidelberg.

\bibitem{Kal}
T.~Kailath.
\newblock {\em Linear systems}.
\newblock Prentice-Hall information and system science series. Prentice-Hall,
  Englewood Cliffs, N.J, 1980.

\bibitem{KMU}
J.~Kurihara, R.~Matsumoto, and T.~Uyematsu.
\newblock Relative generalized rank weight of linear codes and its applications
  to network coding.
\newblock {\em IEEE Transactions on Information Theory}, 61(7):3912--3936,
  2015.

\bibitem{JP20}
J.~Lieb and R.~Pinto.
\newblock Constructions of {MDS} convolutional codes using superregular
  matrices.
\newblock {\em Journal of Algebra Combinatorics Discrete Structures and
  Applications}, 7(1):73--84, 2020.

\bibitem{LPR}
J.~Lieb, R.~Pinto, and J.~Rosenthal.
\newblock Convolutional codes.
\newblock In W.~Cary Huffman, Jon-Lark Kim, and Patrick Sol\'e, editors, {\em
  Concise Encyclopedia of Coding Theory}, pages 197--225. Chapman and Hall/CRC,
  2021.

\bibitem{MPM}
U.~Martínez-Peñas and R.~Matsumoto.
\newblock Relative generalized matrix weights of matrix codes for universal
  security on wire-tap networks.
\newblock {\em IEEE Transactions on Information Theory}, 64(4):2529--2549,
  2017.

\bibitem{OS}
F.~Oggier and A.~Sboui.
\newblock On the existence of generalized rank weights.
\newblock In {\em 2012 International Symposium on Information Theory and its
  Applications}, pages 406--410, 2012.

\bibitem{Rav16}
A.~Ravagnani.
\newblock Generalized weights: an anticode approach.
\newblock {\em J. Pure Appl. Algebra}, 220(5):1946--1962, 2016.

\bibitem{RS}
J.~Rosenthal and R.~Smarandache.
\newblock Maximum {D}istance {S}eparable convolutional codes.
\newblock {\em Applicable Algebra in Engineering, Communication and Computing},
  10:15--32, 1999.

\bibitem{RY}
J.~Rosenthal and E.~V. York.
\newblock On the generalized {H}amming weights of convolutional codes.
\newblock {\em IEEE Transactions on Information Theory}, 43(1):330--335, 1997.

\bibitem{RST}
V.~Tomas, J.~Rosenthal, and R.~Smarandache.
\newblock Decoding of convolutional codes over the erasure channel.
\newblock {\em IEEE Transactions on Information Theory}, 58(1):90--108, 2012.

\bibitem{Wei}
V.~K. Wei.
\newblock Generalized {H}amming weights for linear codes.
\newblock {\em IEEE Transactions on Information Theory}, 37(5):1412--1418,
  1991.

\bibitem{Y}
E.~V. York.
\newblock {\em Algebraic description and construction of error correcting
  codes, a systems theory point of view}.
\newblock PhD thesis, {U}niversity of {N}otre {D}ame, {N}otre {D}ame, {IN},
  1997.

\end{thebibliography}
\end{document}